\documentclass[final,1p,times,twocolumn]{elsarticle}

\usepackage{amssymb}
\usepackage[nointegrals]{wasysym}%% The amsmath package provides various useful equation environments.
\usepackage{amsmath}
\usepackage{xurl}
\usepackage{hyperref}
\usepackage{amsthm}

\newtheorem{theorem}{Theorem}
\newtheorem{corollary}{Corollary}
\newtheorem{property}{Property}

\usepackage{listings}

\journal{Security Journal}

\begin{document}

\begin{frontmatter}

\title{	Argos: A Multi-Tenant SaaS Architecture for Distributed Edge Security Orchestration}

\author[label1]{Agust\'in Javier Di Bartolo\corref{cor1}}
\ead{adibartolo@unex.es}

\author[label1]{Mohammadhossein Homaei}
\author[label1]{Andr\'es Caro}
\author[label1]{Mar \'Avila}

\cortext[cor1]{Corresponding author}

\affiliation[label1]{
    organization={Departamento de Ingenier\'ia de Sistemas Inform\'aticos y Telem\'aticos, Universidad de Extremadura},
    addressline={Av/ Universidad S/N},
    city={C\'aceres},
    postcode={10003},
    state={Extremadura},
    country={Spain}
}
%% Abstract
\begin{abstract}
The growing adoption of distributed edge devices for security auditing and operational monitoring exposes limitations in lightweight control mechanisms traditionally used to manage such nodes. While commercial Command-and-Control (C2) platforms offer advanced orchestration features, their licensing constraints, operational complexity, and proprietary nature limit accessibility for small and medium-sized enterprises and research environments. This paper presents \textit{Argos}, a multi-tenant SaaS C2 architecture designed for transparent, automated, and long-term orchestration of distributed edge devices. The framework implements a structured control plane with role-based access control (RBAC), tenant isolation, persistent task management, and cron-based workflow automation. Unlike operator-driven or stateless execution models, \textit{Argos} enables scheduled and autonomous coordination of security tasks across heterogeneous environments. To formally substantiate its isolation guarantees, the platform is modeled through a set-theoretic authorization framework, and a tenant non-interference property is proven. The proposed architecture was evaluated across three geographically distributed edge devices, executing 810 automated tasks mapped to MITRE ATT\&CK techniques. Results demonstrate a global failure rate of 3.33\% (95\% Wilson confidence interval: [2.30\%, 4.81\%]) with no statistically significant difference across ARM and x86 architectures ($\chi^2 = 1.55$, $p = 0.21$), and stable operation under heterogeneous latency conditions, including intercontinental deployments. To ensure scientific transparency and reproducibility, the evaluation methodology, sanitized datasets, and statistical scripts have been made publicly available in an open-source repository.\end{abstract}


% \begin{highlights}
% \item A multi-tenant SaaS Command-and-Control architecture is proposed for transparent and automated orchestration of distributed security auditing devices.
% \item A set-theoretic authorization model integrating RBAC and tenant isolation is formalized, and a tenant non-interference property is formally proven.
% \item Autonomous, cron-based workflow automation coordinates security tasks across heterogeneous edge devices without operator intervention.
% \item An empirical evaluation of 810 automated executions across three geographically distributed tenants yields a 3.33\% failure rate with architecture-independent behaviour.
% \item To ensure reproducibility, the evaluation framework, sanitized datasets, and statistical scripts are open-sourced and publicly accessible.
% \end{highlights}

%% Keywords

\end{frontmatter}



\section{Introduction}
\label{sec:intro}
Security teams today rely on many small auditing devices placed in different locations, and keeping track of all of them has become harder over time \cite{edgecomputing}. As the number of devices grows, a central point of control is needed to keep task execution consistent and to store audit results in a reliable way. Command and Control (C2) systems were first studied mostly in connection with malware and botnets \cite{zeidanloo2009}, so for a long time they were seen only as attack tools. More recent studies have changed this view. They show that the same C2 ideas can also be used to manage distributed systems in an organized way, through central control, task distribution, and coordinated execution \cite{wang2026sos}. Seen from this angle, a C2 system is not harmful by nature. It becomes a legitimate way to manage devices when it is built and governed for authorized, controlled use.

Within security auditing and experimental validation environments, a C2 based platform can support coordinated management of authorized audit nodes, structured execution of assessment tasks, and persistent traceability of performed actions. Such capabilities are particularly relevant in distributed audit scenarios that require reproducibility, accountability, and centralized policy enforcement. Despite this potential, available solutions remain limited. Commercial platforms are typically distributed under proprietary licenses that restrict independent code auditing and limit adoption in academic or resource constrained environments. Open source tools are largely designed for offensive security operations and adversary emulation, prioritizing stealth and operational flexibility over transparency, governance, and structured multi user management. Academic developments tend to focus on conceptual validation or automated attack campaign execution, without comprehensively addressing multi tenant isolation, structured result persistence, and role based access control within a unified orchestration architecture.

The motivation for this work originates from limitations identified in a prior distributed auditing initiative \cite{DiBartolo2026}, in which the employed C2 platform \cite{telec2} did not support structured coordination across multiple devices nor differentiated management of users and organizational entities. Isolated administration of such devices, whether through remote shell access or independent scripting approaches as observed in earlier implementations \cite{raspy3}, proved insufficient once automated orchestration, traceable result storage, and formal separation between users and organizations were required. These constraints motivated the definition of an architecture explicitly oriented toward the secure orchestration of distributed audit nodes, incorporating role based access control, logical separation between organizations, and centralized persistence of generated artifacts.

In this context, the primary objective of this work is to define the architecture of \textit{Argos}, a role governed C2 platform conceived under a Software as a Service (SaaS) model and intended for deployment in trusted environments. The proposed approach builds on open source components, service segmentation, centralized identity management, and isolated execution within audit nodes. The contribution of this work is twofold. First, it provides the architectural definition and technological integration required to support transparent, multi tenant, and automated orchestration of distributed edge devices, together with a set theoretic formalization of its authorization model and a proven tenant non interference property. Second, it reports an empirical evaluation of the platform across three geographically distributed edge devices, executing 810 automated tasks mapped to MITRE ATT\&CK techniques, and demonstrates stable and architecture independent operation.

The rest of this paper is organized as follows. Section~\ref{sec:related} looks at related work on command and control systems and device orchestration. Section~\ref{sec:arch} explains the architecture and the main components of the proposed platform. Section~\ref{sec:formal} presents a formal model of the authorization system, together with models for reliability, latency, availability, scalability, and threats. Section~\ref{sec:method} describes the experimental setup, including the evaluation environment, the deployment scenario, and the criteria used for evaluation. Section~\ref{sec:results} reports the results and discusses what they mean, along with the limitations found. Finally, Section~\ref{sec:conclusion} closes the paper and points to directions for future work.
\section{Related Work}
\label{sec:related}

\subsection{Command-and-Control Operational Contexts}
\label{subsec2.1}
Command-and-control (C2) platforms are often incorrectly equated 
with basic remote access mechanisms, such as direct shell access 
or manual command execution. While remote access enables 
point-to-point interaction with individual systems, C2 
architectures add a control layer that supports coordinated task 
execution, automation, and centralized visibility across multiple 
devices \cite{haynes2014c2redteam}. This distinction matters for distributed security operations, which require scalability, traceability, and repeatability \cite{wang2026sos, Homaei2024}. Recent studies 
confirm this framing gap persists. Chatzoglou and Kambourakis 
\cite{chatzoglou2025c3} optimize C2 for stealth via a decentralized 
messaging channel. Winkler and Sharma \cite{winkler2025wazuh} 
report only partial SIEM detection of C2 activity, near 85\%. Rao 
et al. \cite{rao2024postexploitation} evaluate C2 frameworks 
strictly as post-exploitation tooling. None addresses governed, 
multi-tenant orchestration.
In practice, C2 solutions serve three broad contexts: commercial 
platforms, open-source frameworks, and research-oriented systems. 
The following subsections analyze each category, identifying 
strengths and limitations that characterize the structural gap 
addressed by the proposed architecture.

\subsection{Commercial Command-and-Control Platforms}
\label{subsec2.2}
Commercial command-and-control platforms are widely adopted by 
large enterprises and professional security service providers due 
to their maturity and comprehensive feature sets. These platforms 
typically provide centralized management interfaces, collaborative 
operator support, advanced payload and session management, and 
extensive integration with post-exploitation and adversary 
simulation workflows \cite{cobaltstrike,alharthi2025adexploit}. 
Their architectures are optimized for reliability and scalability, 
enabling coordinated operations across complex and distributed 
infrastructures.
Despite these advantages, commercial C2 platforms are commonly 
distributed under restrictive licensing and subscription-based 
pricing models. This limits their accessibility for small and medium-sized organizations, academic institutions, and experimental environments \cite{Homaei2025VCD}. Their closed-source nature also restricts 
transparency and verifiability, posing challenges for organizations 
requiring full control, internal auditing, or compliance with 
governance policies. This gap persists. Anand et al. 
\cite{anand2026c2framework} report the same choice between flagged 
open-source tools and proprietary solutions few understand. Their 
own framework still lacks role-based access control and tenant 
isolation. As a result, commercial solutions remain effective for 
large-scale engagements but are often unsuitable for 
resource-constrained or research-oriented contexts.

\subsection{Open-Source Command-and-Control Platforms}
\label{subsec2.3}
Open-source command-and-control platforms provide an accessible 
alternative for organizations and security teams that lack the 
resources to develop proprietary C2 infrastructures. These 
frameworks offer flexible and extensible designs, enabling 
encrypted communications, session handling, and modular command 
execution across distributed systems. Widely used tools such as 
Empire, Sliver, Merlin, and Havoc illustrate the technical 
maturity achieved by community-driven C2 development 
\cite{piet2018opensourcec2}. Maeda and Mimura 
\cite{maeda2021postexploitation} even combine PowerShell Empire 
with deep reinforcement learning to automate post-exploitation.
However, most open-source C2 platforms are primarily designed to 
support offensive security operations and adversary emulation 
scenarios. Consequently, their architectures tend to prioritize 
stealth, evasion, and short-lived engagements over transparency, 
long-term orchestration, and multi-user governance. Additionally, 
the association of these tools with malicious activity can 
introduce reputational and operational risks when deployed in 
controlled or regulated environments. Ward and Wang 
\cite{ward2025sliverdetect} show that Sliver's DNS-based C2 
channel can now be fingerprinted with over 99\% accuracy. 
Table~\ref{tab:opensource_c2_comparison} summarizes 
representative frameworks along these lines, highlighting 
their primary use cases and core characteristics.

\begin{table}[htbp]
\centering
\caption{Comparison of Representative Open-Source Command-and-Control Platforms}
\label{tab:opensource_c2_comparison}
\footnotesize
\begin{tabular}{p{0.9cm} p{0.8cm} p{3.0cm} p{7.1cm}}
\hline
\hline
\textbf{Ref.} & \textbf{Platform} & \textbf{Primary Use} & \textbf{Advantages} \\
\hline
\cite{havoc} &
Havoc &
Red team operations &
Modern GUI-based C2 framework with modular payloads and encrypted communications, designed as an open-source alternative to commercial platforms \\
\cite{piet2018opensourcec2,pupy} &
Pupy &
Cross-platform post-exploitation &
Supports multiple operating systems, emphasizes in-memory execution, and provides a flexible Python-based extension model \\
\cite{piet2018opensourcec2,maeda2021postexploitation,empire} &
Empire &
Post-exploitation automation &
Mature framework with a large PowerShell module ecosystem and strong integration with offensive security workflows \\
\cite{merlin} &
Merlin &
Experimental HTTP-based C2 &
Research-oriented C2 framework leveraging encrypted HTTP/2 communications and novel transport mechanisms \\
\cite{ward2025sliverdetect,sliver} &
Sliver &
Red team and penetration testing &
Lightweight and cross-platform C2 framework written in Go, supporting modern transports and scalable session management \\
\hline
\hline
\end{tabular}
\end{table}

As shown in Table~\ref{tab:opensource_c2_comparison}, open-source 
C2 platforms provide powerful operational capabilities and reduce 
barriers to adoption compared to commercial solutions. Nevertheless, 
these frameworks consistently prioritize offensive flexibility and 
short-term engagements over transparency, multi-user governance, 
and persistent orchestration, restricting their applicability in 
trusted operational environments.

\begin{table}[htbp]
\centering
\caption{Comparison of Research-Oriented Command-and-Control Frameworks}
\label{tab:related_c2_comparison}
\footnotesize
\begin{tabular}{p{0.12cm} p{6.5cm} p{5.5cm}}
\hline
\hline
\textbf{Ref} & \textbf{Principal Advantage} & \textbf{Main Limitation} \\
\hline
\cite{buzzard2024} &
Decentralized post-exploitation C2 architecture leveraging encrypted tunneling without centralized redirectors, enabling low-latency and resilient command-and-control operations. &
Primarily designed for offensive post-exploitation scenarios; lacks multi-user governance and persistent orchestration for trusted environments. \\
\cite{mephisto2024} &
Comprehensive adversary emulation framework supporting structured and repeatable execution of attack behaviors across multiple phases. &
Emphasizes adversarial simulation and stealth; implemented as a research prototype with limited focus on governance and sustained deployment. \\
\cite{deepred2025} &
Integration of machine learning techniques within the C2 control plane to enable intelligent task coordination and automated decision-making. &
High architectural complexity and research-driven orientation limit suitability for practical, long-term operational deployment. \\
\cite{Badea2025} &
Structured automation of adversary emulation workflows, facilitating controlled experimentation and reproducible attack simulation. &
Focused on simulation rather than persistent orchestration, multi-user management, or coordination of trusted distributed devices. \\
\cite{zeidanloo2009} &
Foundational analysis of centralized, decentralized, and peer-to-peer C2 architectures that informs modern system design. &
Conceptual study that does not address automation, governance, or operational orchestration requirements. \\
\cite{laccolith2024} &
Hypervisor-based adversary emulation that evades all five tested commercial anti-virus products, whereas MITRE CALDERA fails to evade detection under the same conditions. &
Purely offensive emulation tool with no multi-tenant, governance, or persistent orchestration layer; validated only as a standalone research prototype. \\
\cite{autonomousadversary2026} &
LLM-based agents plan and orchestrate multi-step red-team activity, such as lateral movement, with an explicit evaluation emphasis on governance. &
Scoped to bounded adversary-emulation scenarios rather than persistent, multi-tenant SaaS orchestration; governance is evaluated, not architecturally enforced. \\
\hline
\hline
\end{tabular}
\end{table}

\subsection{Research-Oriented Command-and-Control Frameworks}
\label{subsec2.4}
Academic research has explored C2 architectures primarily within 
the context of adversarial emulation, architectural 
experimentation, and conceptual validation. In contrast to 
commercial platforms or community-driven offensive tools, 
research-oriented C2 frameworks typically emphasize structural 
innovation, resilience modeling, experimental reproducibility, or 
intelligent coordination mechanisms under controlled laboratory 
conditions. While these contributions advance theoretical 
understanding, they often prioritize conceptual novelty over 
long-term deployment considerations.
Table~\ref{tab:related_c2_comparison} summarizes selected academic 
approaches, outlining their principal strengths and key 
limitations.

As observed in Table~\ref{tab:related_c2_comparison}, 
research-oriented C2 frameworks provide valuable insights into 
architectural diversification and adversary emulation 
methodologies. However, these systems are predominantly 
implemented as experimental prototypes and rarely integrate 
sustained operational capabilities such as formal multi-user 
governance, multi-tenant isolation, persistent task orchestration, 
or deployment across resource-constrained distributed audit nodes. 
Recent work confirms this. Laccolith \cite{laccolith2024} evades 
detection where CALDERA \cite{caldera} fails. LLM-based red-team 
agents \cite{autonomousadversary2026} add coordination but target 
lateral movement, not persistent orchestration.

\subsection{Open Challenges and Research Gaps}
\label{subsec2.5}
Despite the extensive availability of commercial, open-source, and 
research-oriented C2 solutions, several structural challenges 
remain insufficiently addressed. Existing frameworks typically 
prioritize either offensive realism or minimal operational 
complexity, leaving a gap for platforms designed to support 
transparent, multi-user orchestration of trusted distributed 
devices. Even recent proposals confirm this gap, still forcing a 
choice between flagged open-source tools and opaque proprietary 
alternatives \cite{anand2026c2framework}.
Table~\ref{tab:c2_comparison} provides a consolidated comparison 
of representative C2 approaches, highlighting their primary focus 
and inherent trade-offs.

\begin{table}[htbp]
\centering
\caption{High-Level Comparison of Command-and-Control Approaches}
\label{tab:c2_comparison}
\small
\begin{tabular}{p{3.4cm} p{3.5cm} p{5.2cm}}
\hline
\hline
\textbf{Approach} &
\textbf{Primary Focus} &
\textbf{Key Limitations} \\
\hline
Commercial C2 Platforms &
Professional red teaming and adversary simulation &
High licensing cost, closed-source design, limited transparency and accessibility for small organizations \\
Open-Source C2 Platforms &
Offensive security operations and practitioner-driven tooling &
Strong offensive bias, limited governance and transparency, often designed for short-lived engagements \\
Academic C2 Frameworks &
Research experimentation and adversary emulation &
Prototype-level implementations with limited support for long-term operational deployment \\
Argos (Proposed) &
Trusted security orchestration and structured auditing &
Initial implementation scope and evolving feature maturity \\
\hline
\hline
\end{tabular}
\end{table}
The analysis reveals several persistent research and engineering gaps:
\begin{itemize}
  \item \textbf{Limited lifecycle coverage:} Most C2 platforms concentrate on post-exploitation phases and do not support coordinated orchestration across earlier stages such as reconnaissance, validation, and structured assessment workflows.
  
  \item \textbf{Absence of formal role-based interaction models:} Existing systems primarily target technical operators and rarely incorporate structured role-based interfaces suitable for supervisory or organizational oversight.
  
  \item \textbf{Insufficient auditability and traceability:} Persistent, structured logging and centralized evidence retention are often not integral design components, limiting applicability in compliance-oriented environments.
  
  \item \textbf{Restricted automation capabilities:} Native support for structured playbooks, scheduled execution, and repeatable orchestration of security tasks remains limited.
  
  \item \textbf{Weak multi-user and multi-device coordination:} Coordinated management of multiple users and heterogeneous distributed nodes within trusted environments is not comprehensively addressed.
  
  \item \textbf{Limited suitability for long-term edge deployment:} Many existing C2 solutions are not designed for sustained operation on resource-constrained audit devices deployed at the network edge.
  
  \item \textbf{Lack of AI-assisted coordination in governed C2 
  platforms:} Recent work integrates AI agents into bounded 
  adversary-emulation tasks, such as LLM-driven red-team 
  orchestration \cite{autonomousadversary2026} and 
  reinforcement-learning-based post-exploitation 
  \cite{maeda2021postexploitation}. No existing framework, 
  however, embeds AI-assisted coordination within a persistent, 
  multi-tenant, governance-enforcing C2 platform.
  
\end{itemize}
These identified gaps motivate the architectural design of \textit{Argos}, presented in the following section, which aims to bridge the divide between experimental C2 research and production-oriented, governance-aware security orchestration.

\subsection{Capability Comparison}
\label{subsec2.6}
This subsection compares \textit{Argos} with representative
frameworks on capability dimensions relevant to trusted, long-term
orchestration. Each dimension is rated supported, partial, or
absent, from official documentation and repositories rather than
re-execution. This fits since the reference frameworks are
operator-driven offensive tools, while \textit{Argos} is
automation-oriented and multi-tenant; a direct performance
benchmark would not be fair.
The reference set comprises Sliver \cite{sliver} and Havoc
\cite{havoc} (open-source multiplayer frameworks), Empire
\cite{empire} (a mature post-exploitation framework), Mythic
\cite{mythic} (the closest analogue: web-based multi-operator
access, RBAC, container deployment), MITRE Caldera \cite{caldera}
(an ATT\&CK-aligned, auditable emulation platform), and Cryton
\cite{cryton} (an academic framework with cron-based scheduling).

\begin{table*}[htbp]
\centering
\caption{Capability comparison of \textit{Argos} against
representative command-and-control and orchestration frameworks.
$\CIRCLE$ denotes full support, $\LEFTcircle$ partial or
operation-scoped support, and $\Circle$ absence of the capability
based on official documentation.}
\label{tab:capability_comparison}
\footnotesize
\begin{tabular}{p{4.3cm} c c c c c c c}
\hline
\hline
\textbf{Capability dimension} &
\textbf{Sliver} & \textbf{Havoc} & \textbf{Empire} &
\textbf{Mythic} & \textbf{Caldera} & \textbf{Cryton} &
\textbf{Argos} \\
\hline
Multi-operator support        & $\CIRCLE$ & $\CIRCLE$ & $\LEFTcircle$ & $\CIRCLE$ & $\CIRCLE$ & $\LEFTcircle$ & $\CIRCLE$ \\
Role-based access control     & $\Circle$ & $\Circle$ & $\Circle$ & $\LEFTcircle$ & $\LEFTcircle$ & $\Circle$ & $\CIRCLE$ \\
Multi-tenant isolation        & $\Circle$ & $\Circle$ & $\Circle$ & $\Circle$ & $\Circle$ & $\Circle$ & $\CIRCLE$ \\
Scheduled / cron automation   & $\Circle$ & $\Circle$ & $\Circle$ & $\Circle$ & $\LEFTcircle$ & $\CIRCLE$ & $\CIRCLE$ \\
Persistent audit trail        & $\LEFTcircle$ & $\LEFTcircle$ & $\LEFTcircle$ & $\CIRCLE$ & $\CIRCLE$ & $\CIRCLE$ & $\CIRCLE$ \\
Non-offensive design intent   & $\Circle$ & $\Circle$ & $\Circle$ & $\Circle$ & $\LEFTcircle$ & $\LEFTcircle$ & $\CIRCLE$ \\
Containerized on-demand exec. & $\Circle$ & $\Circle$ & $\Circle$ & $\CIRCLE$ & $\LEFTcircle$ & $\LEFTcircle$ & $\CIRCLE$ \\
Edge / ARM deployment focus   & $\LEFTcircle$ & $\Circle$ & $\Circle$ & $\Circle$ & $\Circle$ & $\Circle$ & $\CIRCLE$ \\
Formal isolation guarantee    & $\Circle$ & $\Circle$ & $\Circle$ & $\Circle$ & $\Circle$ & $\Circle$ & $\CIRCLE$ \\
Open-source availability      & $\CIRCLE$ & $\CIRCLE$ & $\CIRCLE$ & $\CIRCLE$ & $\CIRCLE$ & $\CIRCLE$ & $\LEFTcircle$ \\
\hline
\hline
\end{tabular}
\end{table*}
A consistent pattern emerges. Sliver, Havoc, and Empire offer
strong operator collaboration and encrypted communication, but
lack role-based access control, multi-tenant isolation, and native
scheduling, built for short-lived engagements, not sustained
governance. Mythic goes further, with role-scoped operators and
containerized deployment, but its access model is operation-scoped,
not tenant-scoped, with no strict organizational partitioning.
Caldera and Cryton add structured, repeatable automation,
cron-driven for Cryton, but as emulation tools they too omit
tenant isolation and formal governance.
Only \textit{Argos} satisfies multi-tenant isolation and a
formal isolation guarantee, matching the non-interference property
proven in Section~\ref{subsec:formal_isolation}. It pairs the
automation maturity of Cryton and Caldera with a governance and
isolation model absent elsewhere, substantiating the structural
gap motivating \textit{Argos}'s architecture.

\section{System Architecture}
\label{sec:arch}

\textit{Argos} is a multi-user, multi-tenant command-and-control platform delivered under a Software-as-a-Service model. The central infrastructure is hosted in the cloud, centralizing management, orchestration, and data storage. This reduces deployment complexity for client organizations and supports long-term scalability. The platform manages trusted security auditing devices such as penetration testing dropboxes based on single-board computers \cite{artificesecurity}. These devices execute automated and scheduled security assessments across geographically dispersed environments.

The architecture is organized into five functional stages: access and role management, validation and visualization, orchestration and data persistence, secure connectivity, and endpoint execution. This modular design separates responsibilities, enforces component isolation, and supports horizontal scalability. Figure~\ref{fig:architecture} provides a global view of the proposed architecture; each stage is analyzed in the following subsections.

\begin{figure}[ht]
    \centering
    \includegraphics[width=0.8\linewidth]{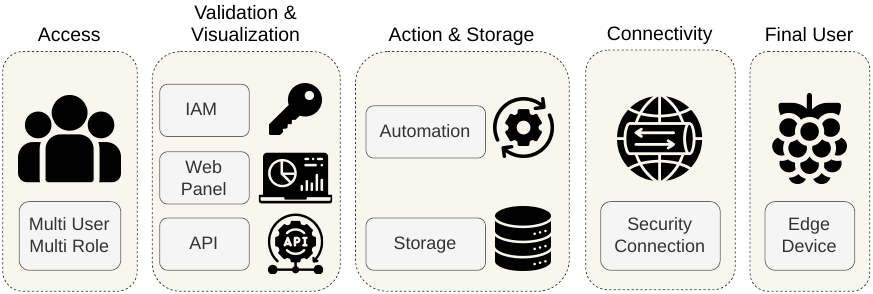}
    \caption{\textit{Argos} C2 platform architecture.}
    \label{fig:architecture}
\end{figure}

The proposed architecture consolidates device management, task coordination, and result aggregation into a single platform. This addresses limitations of isolated C2 deployments, reduces operational complexity, and enables unified visibility across auditing activities. Structured dashboards allow non-technical stakeholders such as supervisors and management personnel to access high-level security insights without interacting with underlying tools.

Transparency and openness are core architectural principles of \textit{Argos}. The platform relies exclusively on open-source components, avoiding opaque third-party services that may introduce hidden behaviors, maintenance risks, or security concerns. Automation logic resides in the control plane, enabling security workflows to execute without continuous operator interaction and minimizing resource consumption on managed endpoints. Task definitions, execution logs, and results are stored centrally, supporting historical analysis, reporting, and internal auditing. These capabilities are typically absent in short-lived C2 deployments.

Unlike traditional C2 systems that rely on malware implants or binary payloads, \textit{Argos} operates exclusively through dedicated and authorized auditing devices (RedPi). These devices are explicitly provisioned to establish secure connections with the control plane and execute predefined tasks, reinforcing the platform's focus on legitimate and transparent operation.

\subsection{Access Management and Role Control}
\label{subsec:access}

The first functional stage of the architecture corresponds 
to user access management and role-based access control. 
\textit{Argos} implements a Role-Based Access Control 
(RBAC) model \cite{zou2010rbac,Sandhu1996}, designed to enforce a 
clear separation of responsibilities and restrict the 
scope of permitted actions for each user type across 
managed edge devices \cite{Bhatt2021}.

As illustrated in Figure~\ref{fig:roles} three principal 
roles are defined within the platform:

\begin{itemize}
  \item \textbf{Administrator:} holds elevated privileges 
  over system configuration, tenant management, user 
  provisioning, and supervision of critical platform 
  parameters. This role has access to structural 
  maintenance and control functions across the system.

  \item \textbf{Operator:} is responsible for the 
  execution and configuration of auditing tasks. Operators 
  may trigger manual actions, define parameters for 
  enabled modules, and schedule recurring tasks, but do 
  not have access to global configuration settings or the 
  administration of other tenants.

  \item \textbf{Read-only:} oriented toward non-technical 
  or supervisory profiles. This role enables visualization 
  of dashboards, historical results, and aggregated 
  metrics, without the ability to execute tasks or modify 
  system configuration.
\end{itemize}

\begin{figure}[h!]
    \centering
    \includegraphics[width=0.70\linewidth]{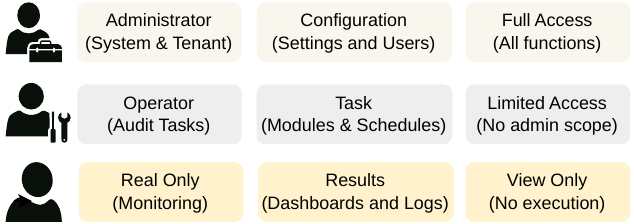}
    \caption{Platform roles}
    \label{fig:roles}
\end{figure}

The platform follows a multi-tenant architectural model 
\cite{kumar2020multitenant}, 
enabling multiple independent organizations to operate 
concurrently within a single shared infrastructure. 

\begin{figure}[h!]
    \centering
    \includegraphics[width=0.6\linewidth]{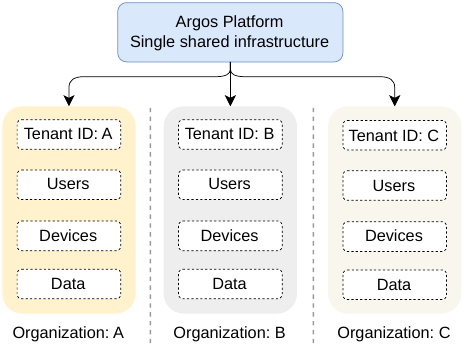}
    \caption{Multi-tenant model}
    \label{fig:multitenant}
\end{figure}

As illustrated in Figure~\ref{fig:multitenant}, each organization 
is represented as an isolated tenant that strictly partitions 
users, devices, tasks, and audit results, preventing any 
visibility or interaction across organizational boundaries. Each 
user belongs to exactly one tenant, so that role and tenant 
together define the authorization boundary governing every 
operation: the role determines \textit{what} a user can do, and 
the tenant determines \textit{with which resources}.

This combined RBAC and multi-tenant model strengthens system 
security while enabling structured collaboration between 
technical and managerial profiles. Different organizational 
levels can interact with the platform in a controlled and 
auditable manner.

\subsection{Validation and Authentication}
\label{subsec:authn}

The second functional stage encompasses identity management, authentication, and user-facing visualization. A dedicated Identity and Access Management (IAM) service centralizes authentication, temporary credential issuance, and authorization policy enforcement \cite{singh2023iam}. Each user accesses only the resources corresponding to their assigned role and organizational context.

The IAM service is deployed as an independent component within the cloud environment, decoupling credential management from the application's core logic. This separation facilitates security policy enforcement, session management, and integration with the role-based access control model defined in the previous stage.

The backend control plane acts as the authorization decision point for every incoming request. It validates the token issued by the IAM service by verifying its signature, expiry, and embedded claims, then applies RBAC policies before granting access to any internal resource. This eliminates per-request communication with the identity service.

\begin{figure}[h!]
    \centering
    \includegraphics[width=1\linewidth]{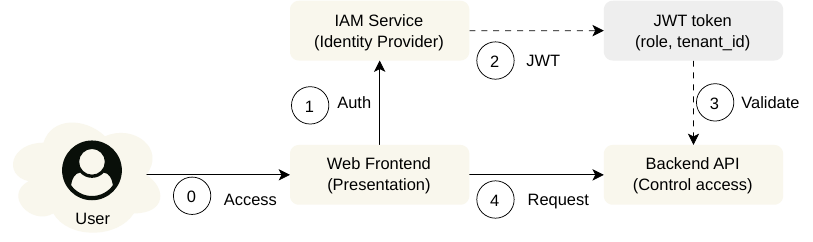}
    \caption{Identity and Access Management (IAM)}
    \label{fig:auth}
\end{figure}

As illustrated in Figure~\ref{fig:auth}, the interaction between 
these components follows four phases:

\begin{enumerate}
\item \textbf{Authentication.} The web frontend delegates authentication to the IAM service, which validates credentials against configured policies and the user directory. 
\item \textbf{Token issuance.} Upon successful validation, the IAM generates a signed JSON Web Token (JWT) containing the user's identity, role, and organizational tenant identifier. The token serves as a temporary access credential for the session. 
\item \textbf{Validation and authorization.} The backend verifies the JWT signature, expiry, and attribute claims. Authentication establishes who the user is; authorization determines what the user may do. The RBAC model is applied accordingly. 
\item \textbf{Controlled access.} Once authorized, the backend grants access to the requested internal services. Responses are returned to the frontend within a view scoped to the user's role and tenant.
\end{enumerate}

The multi-tenant isolation model is enforced 
coordinately between the identity layer and the backend 
control plane. The IAM service embeds the user's tenant 
membership within the issued token, as shown in the 
following example:

\begin{lstlisting}
{
  "sub": "user_id",
  "role": "operator",
  "tenant_id": "org_02"
}
\end{lstlisting}

Each backend request is then processed in function of 
the tenant identifier contained in the token, ensuring 
that data access, task execution, and result 
visualization remain strictly isolated between 
organizations.

\subsection{Orchestration and Data Persistence}
\label{subsec3.4}

Orchestration constitutes the operational core of 
\textit{Argos}, enabling the structured and 
coordinated execution of security auditing tasks across 
distributed edge devices \cite{soar}. Unlike traditional 
point-in-time execution models, the platform integrates 
a dedicated automation engine within the control plane 
that formalizes the sequence of actions through 
structured workflows defined as playbooks 
\cite{bridges2023soar}. This design eliminates the need 
for continuous operator interaction and enables 
repeatable, consistent execution of security routines 
across heterogeneous environments.

Figure~\ref{fig:orchestration_flow} illustrates the complete 
task execution workflow, from the initial user or 
scheduled trigger to the visualization of results in 
the dashboard. The flow can be described through five 
key transitions:

\begin{figure}[htbp]
    \centering
    \includegraphics[width=0.85\linewidth]{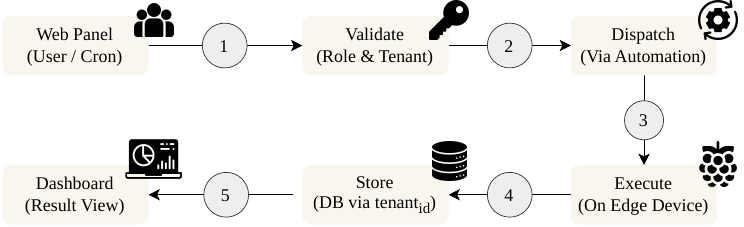}
    \caption{Orchestration and data persistence flow.}
    \label{fig:orchestration_flow}
\end{figure}

\begin{enumerate}
  \item \textbf{Task trigger.} A task is initiated through the 
  web panel either manually by an operator or automatically via 
  a scheduled routine. Both mechanisms produce an identical 
  execution request directed to the backend control plane.
  \item \textbf{Validation.} The backend verifies the user's 
  identity, role, and tenant context, along with the binding 
  between the tenant and the target edge device. Only authorized 
  requests with a valid tenant-device association are forwarded 
  for execution.
  \item \textbf{Task dispatch.} The automation engine generates 
  a unique execution instance and dispatches the task to the 
  assigned edge device over an encrypted tunnel. Playbook logic 
  determines execution paths and parameters based on task type 
  and context.
  \item \textbf{Execution and result collection.} The edge device 
  receives the request through its local agent, instantiates an 
  ephemeral container for the task, and returns structured results 
  asynchronously to the automation engine via webhook.
  \item \textbf{Persistence and visualization.} Results and 
  metadata such as execution status, timestamp, and tenant 
  identifier are stored in the centralized relational database 
  under the corresponding \texttt{tenant\_id}. The dashboard 
  then presents the results scoped to the user's organization.
\end{enumerate}

Figure~\ref{fig:playbook} illustrates an example 
automation playbook used during the evaluation, showing 
how conditional logic and webhook-based interactions 
are combined to orchestrate a complete task execution 
workflow. A formal characterization of the playbook 
structure and of the execution lifecycle, together with 
a termination guarantee, is provided in 
Section~\ref{subsec:formal_playbook}.

\begin{figure}[htbp]
    \centering
    \includegraphics[width=0.85\linewidth]{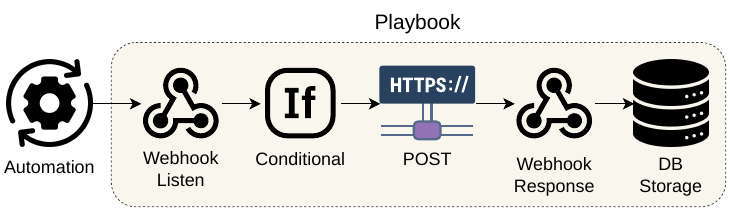}
    \caption{Example automation playbook.}
    \label{fig:playbook}
\end{figure}

The execution model supports both on-demand and 
scheduled tasks, enabling the automation of repetitive 
security routines without requiring constant operator 
intervention, while guaranteeing consistency and 
traceability over time.

Regarding data persistence, all task definitions, 
executions, results, and activity logs are stored in 
a centralized relational database. As shown in 
Figure~\ref{fig:db_tenant}, each record is explicitly 
associated with a \texttt{tenant\_id}, enforcing logical 
isolation between organizations and enabling historical 
analysis, reporting, and auditable traceability of all 
performed operations.

\begin{figure}[htbp]
    \centering
    \includegraphics[width=0.65\linewidth]{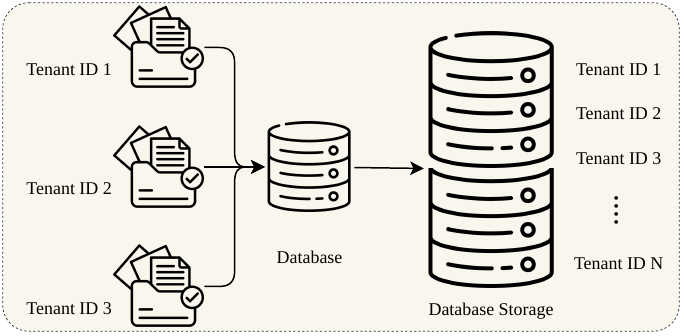}
    \caption{Relational database storage model.}
    \label{fig:db_tenant}
\end{figure}

The separation between orchestration and endpoint 
execution allows edge devices to remain lightweight, 
delegating coordination logic and data storage to the 
cloud-based control plane. This design supports 
horizontal scalability and sustained operation in 
resource-constrained environments. Overall, the 
integration of the automation engine, backend control 
plane, and persistent data layer extends the traditional 
command-and-control model toward a structured 
orchestration scheme with native support for 
multi-user operation, tenant isolation, and repeatable 
task execution.

\subsection{Secure Communication Model}
\label{subsec:comm}

Secure connectivity is a fundamental component of the 
proposed architecture, ensuring protected communication 
between the cloud-hosted control plane and distributed 
edge devices deployed in client environments. Since 
\textit{Argos} operates under a SaaS model, all 
central platform components reside within a controlled 
cloud environment, while edge devices may be located 
in heterogeneous external networks protected by 
firewalls or network address translation mechanisms.

A Web Application Firewall (WAF) is incorporated as 
the first line of defence between external users and 
the platform frontend \cite{waf}. The WAF mitigates 
application-layer attacks such as SQL injection, 
Cross-Site Scripting (XSS), brute-force attempts, and 
anomalous traffic patterns, covering both web panel 
access and backend API requests, thereby reducing the 
public attack surface and reinforcing internal 
authentication and authorization mechanisms.

\begin{figure}[htb]
    \centering
    \includegraphics[width=0.7\linewidth]{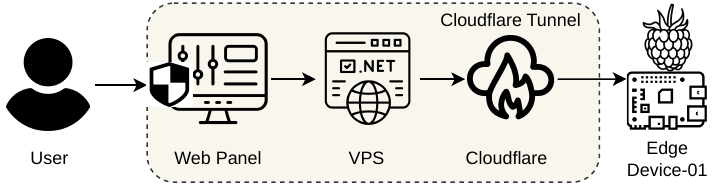}
    \caption{Secure communication model}
    \label{fig:ecos_communication}
\end{figure}

As illustrated in Figure~\ref{fig:ecos_communication}, 
connectivity between the backend control plane and edge devices 
is established through outbound encrypted tunnels initiated from 
the client side \cite{cloudflare_tunnel}. This approach avoids 
exposing inbound services at client environments, reducing the 
external attack surface and simplifying integration in networks 
protected by firewalls or NAT.

Direct administrative access to edge devices is restricted by 
design. Password-based authentication is disabled, and all 
endpoint interactions are mediated through key-based mechanisms 
initiated exclusively by the backend control plane. This ensures 
that actions occur only within authorized platform workflows.

All communication between the backend and edge devices is 
performed via HTTPS requests transmitted through the secure 
tunnel. The backend issues authenticated requests using standard 
methods such as \texttt{POST} and \texttt{GET}, while edge 
devices respond through webhook-based mechanisms. This enables 
asynchronous interaction in which task execution and result 
delivery do not require a persistent connection. 
Table~\ref{tab:api_endpoints} shows representative examples of 
the edge device API interface:

\begin{table}[htbp]
\centering
\caption{Representative Edge Device API Endpoints}
\label{tab:api_endpoints}
\footnotesize
\begin{tabular}{p{1.3cm} p{1.5cm} p{8.9cm}}
\hline
\hline
\textbf{Request} & \textbf{Method} & \textbf{Response (JSON)} \\
\hline
/health & GET & 
\{ "status": "OK" \} \\
/arp\_scan & POST & 
\{ "hosts": [\{ "ip": "...", "mac": "...", 
"vendor": "..." \}], 
"execution\_time": "..." \} \\
\hline
\hline
\end{tabular}
\end{table}

To further protect sensitive operational data, a 
lightweight application-level obfuscation mechanism 
is applied to selected information fields, 
complementing transport-layer encryption while 
remaining computationally efficient for 
resource-constrained devices. Tenant isolation is enforced at the connectivity 
layer. Each edge device is explicitly associated with 
a single tenant, and all requests issued by the 
backend include contextual information previously 
validated by the identity service, preventing any 
cross-tenant interaction between client organizations.

Overall, \textit{Argos}'s secure communication 
model combines WAF-based perimeter protection, 
outbound-only encrypted device connectivity, and 
tenant-scoped communication channels to enable safe, 
scalable, and long-term operation of distributed 
security auditing devices.

\subsection{Endpoint Execution Environment}
\label{subsec:endpoint}

The endpoint execution environment constitutes the 
distributed layer responsible for receiving and 
executing tasks dispatched by the automation engine. 
Unlike traditional models based on persistent implants 
or high-privilege agents, the proposed architecture 
adopts an isolated, on-demand execution approach that 
minimizes the operational footprint on client devices.

Endpoints may be deployed on dedicated physical 
machines, virtual machines, or low-cost auditing 
appliances such as penetration testing dropboxes based 
on single-board computers \cite{artificesecurity}. 
This flexibility allows the platform to adapt to 
diverse operational contexts while preserving a 
consistent execution model. A Linux-based operating 
system is required, and endpoints must support 
container execution technologies and establish secure 
connectivity with the \textit{Argos} control plane 
through encrypted tunnels.

\begin{figure}[h!]
    \centering
    \includegraphics[width=0.6\linewidth]{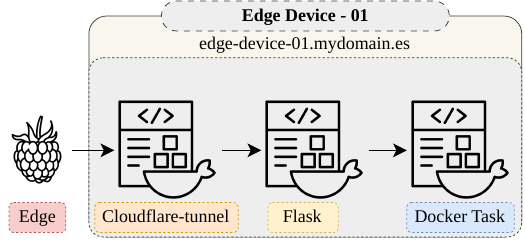}
    \caption{Edge device internal execution flow.}
    \label{fig:edge_execution}
\end{figure}

As illustrated in Figure~\ref{fig:edge_execution}, 
each edge device comprises two core containerized 
components: a secure tunnel service that maintains 
the outbound encrypted channel with the control 
plane, and a lightweight execution agent that 
receives authorized task requests and coordinates 
their execution. The agent exposes a controlled 
webhook interface through which the backend submits 
task requests. Upon receiving a valid request, the 
agent performs the following operations:

\begin{enumerate}
  \item Validates the structure and authenticity of 
  the incoming request.
  \item Dynamically instantiates an ephemeral 
  container associated with the requested task.
  \item Injects the required parameters as 
  environment variables or controlled arguments.
  \item Supervises execution and captures the 
  structured output.
  \item Terminates the container immediately upon 
  task completion.
\end{enumerate}

\begin{figure} [h!]
    \centering
    \includegraphics[width=0.7\linewidth]{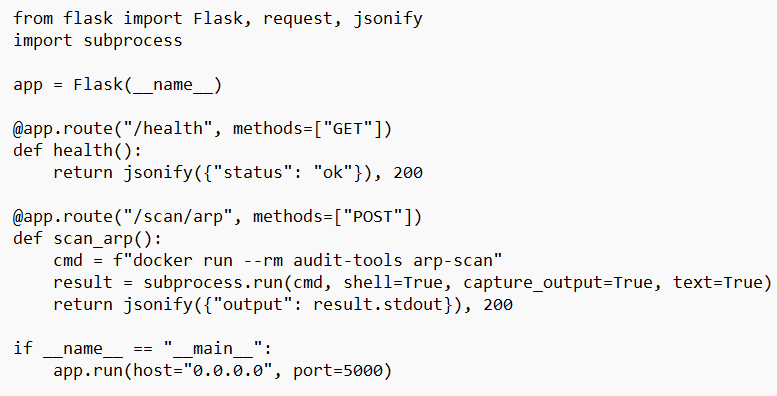}
    \caption{Flask-based execution agent receiving task requests over HTTPS and delegating them to ephemeral containers.}
    \label{fig:flask}
\end{figure}

Figure~\ref{fig:flask} shows a representative 
implementation of the execution agent, illustrating 
how task requests are received via HTTPs endpoints 
and delegated to isolated containers for execution.

From a security perspective, this execution model 
significantly reduces the attack surface of the 
edge device:

\begin{itemize}
  \item Auditing tools are not installed directly 
  on the host operating system.
  \item No services are exposed toward the client's 
  local network.
  \item All tool dependencies are encapsulated 
  within container images.
  \item Executions do not maintain persistent 
  sessions on the host.
\end{itemize}

This design simplifies node maintenance and allows tools or task behaviors to be updated by modifying container images, without requiring direct intervention in the client environment. Furthermore, by delegating heavy coordination and persistence logic to the cloud, the architecture accommodates the severe computational and latency constraints prevalent in SME edge deployments \cite{homaei2026mosaic}. Each device also implements a periodic heartbeat mechanism through which it reports its operational status to the control 
plane. This signal enables detection of connectivity failures, 
tunnel disruptions, or local execution issues, reinforcing 
centralized supervision without requiring direct access to client 
environments.

Through hardened provisioning, containerized on-demand execution, 
and controlled task mediation, the endpoint execution environment 
of \textit{Argos} enables safe, repeatable, and low-impact 
security auditing. This makes the platform particularly suited 
for long-term deployments where stability, isolation, and minimal 
operational overhead are essential.

\section{Formal Model}
\label{sec:formal}

The architectural description in Section~\ref{sec:arch} specifies the
functional behaviour of \textit{Argos} but leaves its central
security claim, namely the strict separation of resources across
organizational boundaries, expressed only in narrative form. This
section introduces a set-theoretic model of the platform's
authorization mechanism that unifies the role-based and multi-tenant
control dimensions into a single decision predicate. On this basis, a
tenant non-interference property is stated and formally proven,
providing the analytical counterpart to the empirical isolation
results reported in Section~\ref{sec:results}. The operational models
governing reliability, latency, availability, and scalability are
developed in the subsequent subsections.

\subsection{System Model and Notation}
\label{subsec:formal_model}

The platform is modeled as an access-controlled system operating over
a finite collection of principals, resources, and operations. Let $U$
denote the set of users, $R$ the set of roles, $T$ the set of tenants,
$D$ the set of registered edge devices, and $O$ the set of operations
that may be requested through the control plane. Let $\mathcal{X}$
denote the set of protected objects, comprising both edge devices and
data resources such as tasks, execution results, logs, and dashboard
views, so that $D \subseteq \mathcal{X}$.

Four total functions bind principals and objects to their governing
attributes, as defined in Equation~\eqref{eq:mappings}:
\begin{equation}
\label{eq:mappings}
\rho : U \rightarrow R, \quad
\tau : U \rightarrow T, \quad
\theta : D \rightarrow T, \quad
\kappa : \mathcal{X} \setminus D \rightarrow T .
\end{equation}
The role assignment $\rho$ maps each user to exactly one role in
$R = \{\textsf{admin}, \textsf{operator}, \textsf{readonly}\}$,
reflecting the RBAC model of Section~\ref{subsec:access}. The membership
function $\tau$ assigns each user to a single tenant, in accordance
with the constraint that every user belongs to exactly one
organization. The functions $\theta$ and $\kappa$ associate each edge
device and each data resource, respectively, with the tenant that owns
it. Totality of these functions formalizes the invariant that no
principal or object may exist outside a tenant scope.

Permitted operations are captured by a permission assignment relation,
given in Equation~\eqref{eq:pa}:
\begin{equation}
\label{eq:pa}
PA \subseteq R \times O ,
\end{equation}
where $(r, o) \in PA$ indicates that role $r$ is authorized to perform
operation $o$. The relation $PA$ is defined administratively and
encodes, for instance, that the \textsf{readonly} role is excluded from
any task-execution or configuration operation.

To reason uniformly over both devices and data resources, a
tenant-scope function $\sigma : \mathcal{X} \rightarrow T$ unifies
$\theta$ and $\kappa$, as expressed in Equation~\eqref{eq:sigma}:
\begin{equation}
\label{eq:sigma}
\sigma(x) =
\begin{cases}
\theta(x), & x \in D, \\[2pt]
\kappa(x), & x \in \mathcal{X} \setminus D .
\end{cases}
\end{equation}
The function $\sigma$ therefore returns the owning tenant of any
protected object, whether it is a physical edge device or a logical
data record scoped by \texttt{tenant\_id}.

\subsection{Authorization Predicate}
\label{subsec:formal_auth}

A request submitted to the control plane is represented as a triple
$q = (u, o, x)$, where $u \in U$ is the requesting user, $o \in O$ is
the requested operation, and $x \in \mathcal{X}$ is the target object.
The backend acts as the single authorization decision point, as
described in Section~\ref{subsec:authn}. Its decision is defined by the
authorization predicate of Equation~\eqref{eq:auth}:
\begin{equation}
\label{eq:auth}
\mathrm{Auth}(u, o, x) \;\Longleftrightarrow\;
\underbrace{\bigl(\rho(u), o\bigr) \in PA}_{\text{role permission}}
\;\wedge\;
\underbrace{\tau(u) = \sigma(x)}_{\text{tenant binding}} .
\end{equation}
Equation~\eqref{eq:auth} makes explicit that authorization is granted
only when both dimensions hold simultaneously. The first conjunct
enforces the RBAC constraint, determining \textit{what} a user may do;
the second conjunct enforces the multi-tenant constraint, determining
\textit{with which resources} the action may be performed. These two
dimensions correspond precisely to the two-dimensional authorization
boundary introduced informally in Section~\ref{subsec:access}. A request is
executed by the control plane if and only if
$\mathrm{Auth}(u, o, x) = \top$.

\subsection{Tenant Isolation Guarantee}
\label{subsec:formal_isolation}

The tenant binding conjunct of Equation~\eqref{eq:auth} is the
structural mechanism that prevents any interaction across
organizational boundaries. This behaviour is formalized as the
following non-interference property.

\begin{theorem}[Tenant Non-Interference]
\label{thm:noninterference}
Let $u \in U$ be an arbitrary user and $x \in \mathcal{X}$ an arbitrary
protected object such that $\sigma(x) \neq \tau(u)$. Then, for every
operation $o \in O$,
\begin{equation}
\label{eq:noninterference}
\mathrm{Auth}(u, o, x) = \bot .
\end{equation}
\end{theorem}

\begin{proof}
By Equation~\eqref{eq:auth}, $\mathrm{Auth}(u, o, x)$ is the logical
conjunction of the role-permission term $(\rho(u), o) \in PA$ and the
tenant-binding term $\tau(u) = \sigma(x)$. By hypothesis,
$\sigma(x) \neq \tau(u)$, hence the tenant-binding term evaluates to
$\bot$. Since conjunction with $\bot$ yields $\bot$ irrespective of the
truth value of the remaining term, it follows that
$\mathrm{Auth}(u, o, x) = \bot$ for every $o \in O$, which establishes
Equation~\eqref{eq:noninterference}.
\end{proof}

Theorem~\ref{thm:noninterference} states that no user can perform any
operation on an object belonging to a different tenant, regardless of
that user's role. In particular, the guarantee holds even for the
\textsf{admin} role, whose elevated privileges satisfy the first
conjunct of Equation~\eqref{eq:auth} but cannot override the tenant
binding. The result therefore captures, as a formal property, the
cross-tenant access checks and cross-edge visibility checks that
returned zero violations in the empirical evaluation of
Section~\ref{sec:results}.

A direct consequence concerns the confidentiality of collected audit
data across organizations.

\begin{corollary}[Result Isolation]
\label{cor:result_isolation}
For any two distinct tenants $t_1, t_2 \in T$ with $t_1 \neq t_2$, no
user $u$ with $\tau(u) = t_1$ can read any data resource
$x \in \mathcal{X} \setminus D$ with $\kappa(x) = t_2$.
\end{corollary}

\begin{proof}
Let $o_{\text{read}} \in O$ denote a read operation on a data resource.
Since $\kappa(x) = t_2$ and $x \notin D$, Equation~\eqref{eq:sigma}
gives $\sigma(x) = t_2$. Because $\tau(u) = t_1 \neq t_2 = \sigma(x)$,
the hypothesis of Theorem~\ref{thm:noninterference} is satisfied, and
therefore $\mathrm{Auth}(u, o_{\text{read}}, x) = \bot$. As read access
is denied at the authorization decision point, no cross-tenant data
disclosure is possible.
\end{proof}

Corollary~\ref{cor:result_isolation} provides the analytical basis for
the result-scoping behaviour validated across the three organizational
deployments, in which execution histories, task outcomes, and
aggregated views remained strictly confined to their originating
tenant.

\subsection{Reliability Model}
\label{subsec:formal_reliability}

Beyond the isolation guarantees, the operational value of the platform
depends on the reliability of its task-execution mechanism. Each
scheduled execution dispatched by the automation engine is modeled as
an independent Bernoulli trial with an unknown success probability
$p \in [0, 1]$, where success denotes the correct completion of a task
and the return of a structured result to the control plane. For a
deployment comprising $|D|$ edge devices, each executing a fixed set of
$K$ daily tasks over an evaluation horizon of $\Delta$ days, the total
number of executions is given by Equation~\eqref{eq:trials}:
\begin{equation}
\label{eq:trials}
n = |D| \cdot K \cdot \Delta .
\end{equation}
Letting the random variable $S$ denote the number of successful
executions, $S$ follows a binomial distribution
$S \sim \mathrm{Bin}(n, p)$, and the maximum-likelihood estimator of
the success probability is the observed success ratio, as stated in
Equation~\eqref{eq:phat}:
\begin{equation}
\label{eq:phat}
\hat{p} = \frac{S}{n} .
\end{equation}
Because the estimated success probability is close to unity, the
symmetric Wald approximation is inaccurate, as it may yield an interval
exceeding the admissible range $[0, 1]$. The Wilson score interval is
therefore adopted, providing reliable coverage for proportions near the
boundary. For a confidence level $1 - \alpha$ with standard-normal
quantile $z = z_{1 - \alpha/2}$, the interval is defined by
Equation~\eqref{eq:wilson}:
\begin{equation}
\label{eq:wilson}
\mathrm{CI}_{1-\alpha}(p) =
\frac{
\hat{p} + \dfrac{z^{2}}{2n}
\;\pm\;
z \sqrt{\dfrac{\hat{p}(1 - \hat{p})}{n} + \dfrac{z^{2}}{4n^{2}}}
}{
1 + \dfrac{z^{2}}{n}
} .
\end{equation}
Equation~\eqref{eq:wilson} is applied in Section~\ref{sec:results} to
accompany the point estimate of the failure rate with a rigorous
interval characterization of execution reliability.

To assess whether the underlying hardware architecture influences
reliability, the executions are partitioned into two groups according
to processor family, ARM and x86. Let $O_{ij}$ denote the observed
count of outcome $j \in \{\text{success}, \text{failure}\}$ within
architecture group $i$, and let $E_{ij}$ denote the corresponding
count expected under the null hypothesis of architecture-independent
reliability. Pearson's chi-squared statistic is given by
Equation~\eqref{eq:chi2}:
\begin{equation}
\label{eq:chi2}
\chi^{2} = \sum_{i} \sum_{j}
\frac{\left(O_{ij} - E_{ij}\right)^{2}}{E_{ij}},
\qquad
E_{ij} = \frac{R_i \, C_j}{n} ,
\end{equation}
where $R_i$ and $C_j$ are the marginal totals of group $i$ and outcome
$j$, respectively. The statistic of Equation~\eqref{eq:chi2}, evaluated
with one degree of freedom, is used in Section~\ref{sec:results} to test the
hardware-agnosticism claim of the containerized execution model.

\subsection{Latency Independence of Execution Correctness}
\label{subsec:formal_latency}

A distinctive characteristic of the communication model described in
Section~\ref{subsec:comm} is its reliance on stateless HTTPS
request-response exchanges rather than persistent interactive sessions.
This subsection formalizes the resulting independence between network
latency and execution correctness, which underlies the platform's
suitability for geographically dispersed deployment.

Consider a single task dispatched from the control plane to an edge
device. Let $L$ denote the one-way network latency, $T_{\text{exec}}$
the local execution time on the device, and $T_{\text{ovh}}$ the
fixed processing overhead. The end-to-end completion time perceived by
the control plane is expressed by Equation~\eqref{eq:ttotal}:
\begin{equation}
\label{eq:ttotal}
T_{\text{total}} = 2L + T_{\text{exec}} + T_{\text{ovh}} .
\end{equation}
The result returned by the edge device is a deterministic function of
the requested task and the observable state of the target environment
at execution time, as expressed by Equation~\eqref{eq:result}:
\begin{equation}
\label{eq:result}
\mathcal{R} = f(\text{task}, \text{target state}) .
\end{equation}
Crucially, the right-hand side of Equation~\eqref{eq:result} contains
no dependence on $L$. Provided that the round-trip time does not exceed
the configured request timeout $T_{\text{to}}$, the latency affects
only the completion time in Equation~\eqref{eq:ttotal} and never the
computed result. This invariance is stated as the following property.

\begin{property}[Latency-Invariant Correctness]
\label{prop:latency}
For any two latency values $L_1, L_2$ satisfying
$2L_1 < T_{\text{to}}$ and $2L_2 < T_{\text{to}}$, the execution
results coincide, as expressed by Equation~\eqref{eq:latency_inv}:
\begin{equation}
\label{eq:latency_inv}
\mathcal{R}(L_1) = \mathcal{R}(L_2) .
\end{equation}
\end{property}

\noindent
Property~\ref{prop:latency} follows immediately from the absence of
$L$ in Equation~\eqref{eq:result}: since the result depends solely on
the task specification and the target state, varying the latency below
the timeout threshold cannot alter the outcome. This property explains
the empirical observation in Section~\ref{sec:results} that intercontinental
deployments, despite exhibiting higher round-trip times, preserve
execution correctness without instability.

\subsection{Availability and Failure Detection}
\label{subsec:formal_availability}

Centralized supervision relies on the periodic heartbeat mechanism
introduced in Section~\ref{subsec:endpoint}, through which each edge device
reports its operational status at a fixed interval $\tau_h$. The
availability of a device is modeled as a two-state continuous-time
Markov process alternating between an operational state and a failed
state, with failure rate $\lambda$ and recovery rate $\mu$. The
steady-state availability is given by Equation~\eqref{eq:availability}:
\begin{equation}
\label{eq:availability}
A = \frac{\mu}{\lambda + \mu}
  = \frac{\mathrm{MTBF}}{\mathrm{MTBF} + \mathrm{MTTR}} ,
\end{equation}
where $\mathrm{MTBF} = 1/\lambda$ is the mean time between failures and
$\mathrm{MTTR} = 1/\mu$ is the mean time to recovery. Equation~\eqref{eq:availability}
expresses availability as a function of the reliability of the client
environment, consistent with the observation in Section~\ref{sec:results}
that most failures originated from external factors such as power and
connectivity interruptions.

Under the periodic reporting scheme, a failure occurring at an
arbitrary instant within a reporting interval is detected at the next
scheduled heartbeat. Assuming the failure onset is uniformly
distributed over the interval, the expected detection delay is given by
Equation~\eqref{eq:detection}:
\begin{equation}
\label{eq:detection}
\mathbb{E}[t_{\text{detect}}] = \frac{\tau_h}{2},
\qquad
t_{\text{detect}} \leq \tau_h .
\end{equation}
Equation~\eqref{eq:detection} bounds the worst-case detection latency
by a single reporting interval, providing a deterministic guarantee on
the timeliness of the unavailability alerts generated by the control
plane.

\subsection{Scalability Bound}
\label{subsec:formal_scalability}

Finally, the sustainable scale of a single control-plane instance is
characterized in terms of the aggregate request rate it must process.
For a deployment of $N = |D|$ edge devices, the control plane receives
two classes of periodic traffic: heartbeat reports and scheduled task
results. The heartbeat arrival rate is given by
Equation~\eqref{eq:hb_rate}:
\begin{equation}
\label{eq:hb_rate}
\Lambda_{\text{hb}} = \frac{N}{\tau_h} ,
\end{equation}
while the aggregate rate of scheduled task executions, for $K$ daily
tasks per device distributed over a period of $86{,}400$ seconds, is
given by Equation~\eqref{eq:task_rate}:
\begin{equation}
\label{eq:task_rate}
\Lambda_{\text{task}} = \frac{N \cdot K}{86{,}400} .
\end{equation}
Let $C$ denote the maximum request-processing capacity of the control
plane, expressed in requests per second. The stability condition
requires that the aggregate offered load remain below capacity,
$\Lambda_{\text{hb}} + \Lambda_{\text{task}} \leq C$. Substituting
Equations~\eqref{eq:hb_rate} and \eqref{eq:task_rate} and solving for
$N$ yields the admissible device bound of
Equation~\eqref{eq:nmax}:
\begin{equation}
\label{eq:nmax}
N_{\max} =
\left\lfloor
\frac{C}{\dfrac{1}{\tau_h} + \dfrac{K}{86{,}400}}
\right\rfloor .
\end{equation}
Equation~\eqref{eq:nmax} provides an analytical estimate of the number
of edge devices a single control-plane instance can sustain as a
function of its processing capacity and the configured heartbeat and
scheduling parameters. This bound frames the scalability limits
discussed in Section~\ref{subsec:limitations} and motivates the horizontal
scaling considerations identified as future work in
Section~\ref{sec:conclusion}.

\subsection{Playbook Formalization}
\label{subsec:formal_playbook}

The automation engine described in Section~\ref{subsec3.4} coordinates
each security workflow through a playbook, whose informal
representation is given in Figure~\ref{fig:playbook}. To reason
precisely about the execution structure and termination of these
workflows, a playbook is modeled at two complementary levels: its
static task-dependency structure, captured as a directed acyclic
graph, and the dynamic lifecycle of a single execution, captured as a
deterministic finite automaton.

\subsubsection{Structural Model}
\label{subsubsec:playbook_dag}

The static structure of a playbook is represented as the directed
acyclic graph defined in Equation~\eqref{eq:playbook_dag}:
\begin{equation}
\label{eq:playbook_dag}
P = (V, E), \qquad
E \subseteq V \times V,
\end{equation}
where $V = \{ v_1, v_2, \dots, v_m \}$ is the finite set of playbook
steps, each corresponding to an atomic action such as a task dispatch,
a conditional branch, or a persistence operation, and an edge
$(v_i, v_j) \in E$ denotes that step $v_j$ may execute only after step
$v_i$ has completed. The acyclicity constraint of
Equation~\eqref{eq:dag_acyclic} forbids cyclic dependencies among
steps:
\begin{equation}
\label{eq:dag_acyclic}
\nexists\, (v_{i_1}, v_{i_2}, \dots, v_{i_k}) : \;
(v_{i_j}, v_{i_{j+1}}) \in E \;\wedge\; v_{i_1} = v_{i_k} .
\end{equation}
Acyclicity guarantees that a playbook admits at least one topological
ordering of its steps, which the automation engine follows during
execution. Conditional logic is represented by branch nodes whose
outgoing edges are labeled by mutually exclusive guard predicates: for
a branch node $v_i$ with successors $\{ v_{j_1}, \dots, v_{j_d} \}$,
the associated guards $g_{j_1}, \dots, g_{j_d}$ satisfy the
completeness and exclusivity condition of
Equation~\eqref{eq:guards}:
\begin{equation}
\label{eq:guards}
\bigvee_{l=1}^{d} g_{j_l} = \top
\qquad \text{and} \qquad
g_{j_a} \wedge g_{j_b} = \bot \;\; \forall\, a \neq b ,
\end{equation}
ensuring that exactly one outgoing branch is selected at runtime,
so that the execution path through $P$ is deterministic once the
observable state is fixed.

\subsubsection{Execution Lifecycle Model}
\label{subsubsec:playbook_dfa}

While the DAG captures the dependency structure, the progression of an
individual execution instance is modeled as the deterministic finite
automaton of Equation~\eqref{eq:playbook_dfa}:
\begin{equation}
\label{eq:playbook_dfa}
M = (Q, \Sigma, \delta, q_0, F),
\end{equation}
whose components are instantiated as follows. The state set
$Q = \{ q_0, q_{\mathrm{val}}, q_{\mathrm{disp}}, q_{\mathrm{exec}},
q_{\mathrm{pers}}, q_{\mathrm{done}}, q_{\mathrm{fail}} \}$ corresponds
respectively to the idle, validated, dispatched, executing, persisted,
completed, and failed states of an execution. The input alphabet
$\Sigma = \{ \mathsf{trigger}, \mathsf{authorize}, \mathsf{ack},
\mathsf{result}, \mathsf{store}, \mathsf{error}, \mathsf{timeout} \}$
represents the events observed by the control plane. The initial state
is $q_0$, and the set of accepting states is
$F = \{ q_{\mathrm{done}}, q_{\mathrm{fail}} \}$, reflecting that every
execution terminates in either a completed or a failed condition. The
transition function $\delta$ follows the nominal path of
Equation~\eqref{eq:dfa_path}:
\begin{equation}
\label{eq:dfa_path}
q_0 \xrightarrow{\mathsf{trigger}} q_{\mathrm{val}}
\xrightarrow{\mathsf{authorize}} q_{\mathrm{disp}}
\xrightarrow{\mathsf{ack}} q_{\mathrm{exec}}
\xrightarrow{\mathsf{result}} q_{\mathrm{pers}}
\xrightarrow{\mathsf{store}} q_{\mathrm{done}} ,
\end{equation}
where each transition corresponds to one of the five orchestration
stages enumerated in Section~\ref{subsec3.4}, namely task trigger,
validation, dispatch, execution with result collection, and
persistence. Any $\mathsf{error}$ or $\mathsf{timeout}$ event observed
in a non-accepting state $q \in Q \setminus F$ drives the automaton to
the failed state, as expressed by Equation~\eqref{eq:dfa_fail}:
\begin{equation}
\label{eq:dfa_fail}
\delta(q, e) = q_{\mathrm{fail}}, \qquad
\forall\, q \in Q \setminus F, \;\;
e \in \{ \mathsf{error}, \mathsf{timeout} \} .
\end{equation}
Because $Q$ is finite and every $\mathsf{error}$ or $\mathsf{timeout}$
transition leads directly to the absorbing state $q_{\mathrm{fail}}$
while the nominal path of Equation~\eqref{eq:dfa_path} is strictly
progressing, the automaton reaches an accepting state in at most
$|Q| - 1$ transitions. This yields the termination guarantee stated in
Property~\ref{prop:termination}.

\begin{property}[Playbook Termination]
\label{prop:termination}
For any playbook whose structure satisfies
Equations~\eqref{eq:playbook_dag}--\eqref{eq:guards}, every execution
instance modeled by the automaton $M$ of
Equation~\eqref{eq:playbook_dfa} reaches an accepting state in
$F = \{ q_{\mathrm{done}}, q_{\mathrm{fail}} \}$ within a bounded
number of transitions.
\end{property}

\noindent
Property~\ref{prop:termination} formalizes the operational observation
of Section~\ref{sec:results} that every one of the 810 scheduled
executions terminated in a definite outcome, either a recorded success
or a logged failure, with no execution left in an indeterminate state.
The DAG constraint of Equation~\eqref{eq:dag_acyclic} guarantees that
the dependency structure cannot induce non-termination, while the
absorbing failure state of Equation~\eqref{eq:dfa_fail} guarantees that
adverse events, such as the client-side tunnel disruptions reported in
Section~\ref{sec:results}, are mapped to a well-defined terminal
outcome rather than to an undefined execution state. Together, the
structural and lifecycle models provide a rigorous account of the
determinism, auditability, and termination of the platform's automated
workflows.

\subsection{Threat Model}
\label{subsec:formal_threat}

The preceding models characterize the correctness and reliability of
\textit{Argos} under benign operating conditions. This subsection
makes explicit the adversarial assumptions under which the security
properties of the platform are intended to hold, following the
standard practice of defining protected assets, adversary
capabilities, trust boundaries, and the resulting security objectives.

\subsubsection{Assets and Trust Boundaries}
\label{subsubsec:assets}

Four classes of protected assets are considered: (i) the
confidentiality and integrity of tenant-scoped audit data, comprising
task definitions, execution results, and logs identified by
\texttt{tenant\_id}; (ii) the authenticity of task-dispatch commands
issued by the control plane to edge devices; (iii) the credentials and
tokens governing authentication and authorization, in particular the
signed JSON Web Tokens introduced in Section~\ref{subsec:authn}; and
(iv) the availability of the orchestration service itself.

The system is partitioned into three trust domains. The
\emph{cloud control plane}, comprising the backend, the identity and
access management service, the automation engine, and the relational
database, is considered trusted and operates within a controlled
environment. Each \emph{edge device} constitutes a semi-trusted domain:
it is explicitly provisioned and authorized, yet deployed within a
client network outside the direct physical control of the platform
operator. The \emph{public network} interconnecting these domains is
considered fully untrusted. The boundaries between these domains
define the surfaces at which the adversary is assumed to operate.

\subsubsection{Adversary Model}
\label{subsubsec:adversary}

The network adversary is modeled according to the Dolev--Yao
formulation \cite{dolevyao1983}, and is assumed to exercise full
control over the untrusted public network. Formally, the adversary
capabilities are the set defined in
Equation~\eqref{eq:adversary}:
\begin{equation}
\label{eq:adversary}
\mathcal{A} = \{\, \textsf{eavesdrop},\ \textsf{intercept},\
\textsf{replay},\ \textsf{inject},\ \textsf{spoof} \,\} .
\end{equation}
Accordingly, the adversary may observe, delay, drop, replay, or forge
messages traversing the public network, and may attempt to impersonate
either a legitimate user or an edge device. In addition, a
\emph{malicious tenant} is considered: an authenticated principal
$u$ with a valid role assignment $\rho(u)$ and tenant membership
$\tau(u)$ who attempts to access resources belonging to a distinct
tenant, thereby probing the isolation boundary formalized in
Section~\ref{subsec:formal_isolation}.

Three categories of threat are explicitly placed outside the scope of
this model. First, a full compromise of the cloud control plane is
excluded, as it constitutes the root of trust of the architecture.
Second, physical compromise of the internal hardware of an edge device,
including extraction of provisioned key material, is not addressed;
the device is assumed to preserve the integrity of its outbound tunnel
service. Third, denial-of-service attacks that saturate the client
network are acknowledged as an availability concern but are not
countered by the present design, since connectivity quality at the
client side lies outside the platform's control, as observed
empirically in Section~\ref{sec:results}.

\subsubsection{Security Objectives and Mitigations}
\label{subsubsec:objectives}

Under the adversary of Equation~\eqref{eq:adversary}, the platform is
designed to uphold the security objectives summarized in
Table~\ref{tab:threat_mitigation}, each mapped to the architectural
mechanism that enforces it. Confidentiality and integrity of data in
transit rely on the encrypted outbound tunnel and transport-layer
security described in Section~\ref{subsec:comm}, which neutralize the
\textsf{eavesdrop} and \textsf{intercept} capabilities. Message
authenticity is enforced through key-based mediation and the
disabling of password authentication on edge devices, countering
\textsf{inject} and \textsf{spoof}. Token integrity is guaranteed by
the digital signature of the JWT, whose verification by the backend
detects any forgery or tampering. Crucially, the malicious-tenant
threat is neutralized not by a perimeter control but by the
authorization predicate itself: Theorem~\ref{thm:noninterference}
guarantees that any cross-tenant request evaluates to $\bot$
regardless of the requester's role, so that isolation holds as a
structural property rather than a configurable policy.

\begin{table}[htbp]
\centering
\caption{Mapping of adversary capabilities to security objectives and
enforcing mechanisms}
\label{tab:threat_mitigation}
\footnotesize
\begin{tabular}{p{3cm} p{4cm} p{5.0cm}}
\hline
\hline
\textbf{Adversary capability} &
\textbf{Security objective} &
\textbf{Enforcing mechanism} \\
\hline
Eavesdrop, Intercept &
Data confidentiality and integrity in transit &
Outbound encrypted tunnel and TLS
(Sec.~\ref{subsec:comm}) \\
Inject, Spoof &
Command authenticity &
Key-based mediation; password authentication
disabled on endpoints \\
Replay &
Freshness of requests &
Token expiry and signed, time-bounded JWT claims
(Sec.~\ref{subsec:authn}) \\
Spoof (external user) &
Perimeter protection &
Web Application Firewall filtering malicious
traffic (Sec.~\ref{subsec:comm}) \\
Malicious tenant &
Tenant isolation &
Authorization predicate
Eq.~\eqref{eq:auth} and
Theorem~\ref{thm:noninterference} \\
\hline
\hline
\end{tabular}
\end{table}

Table~\ref{tab:threat_mitigation} therefore aligns each adversarial
capability with a concrete architectural countermeasure, and grounds
the tenant-isolation objective in the formal guarantee of
Section~\ref{subsec:formal_isolation}. The empirical validation of
these objectives, in particular the zero cross-tenant violations
observed across 330 controlled operations, is reported in
Section~\ref{sec:results}.

\section{Methodology}
\label{sec:method}

\subsection{Evaluation Environment}
\label{subsec4.1}

The evaluation of \textit{Argos} was conducted 
across two sequential deployment phases. During 
the initial phase, the control plane was deployed 
on a locally managed infrastructure, enabling 
rapid iteration and functional validation. 
Following successful validation, the platform 
was migrated to a production-grade cloud 
environment, allowing the maintenance of 
independent development and production 
environments throughout the evaluation period.

To ensure the reproducibility of our empirical findings and maintain scientific transparency, all sanitized execution logs, statistical analysis scripts, and the evaluation framework data utilized in this study have been open-sourced. The complete dataset and methodology artifacts are publicly available in the project's GitHub repository at \url{https://github.com/adibartolo/Argos}.

\begin{figure}[h!]
        \centering 
        \includegraphics[width=0.7\linewidth]{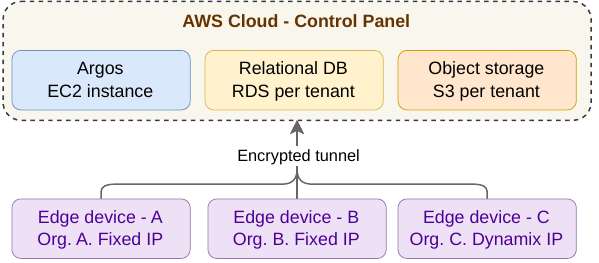}
        \caption{Production environment}
        \label{fig:environment}
    \end{figure}

As illustrated in Figure~\ref{fig:environment} the cloud-hosted control plane was deployed on 
an Amazon Web Services EC2 instance \cite{synergy2026}, providing a managed 
relational database (RDS) for persistent 
tenant-scoped storage and an object storage service 
(S3) for audit artifact retention.

Three edge devices were deployed across 
independent client environments, each assigned 
to a separate tenant, operating under 
heterogeneous network conditions including both 
fixed and dynamic IP addressing. Connectivity 
was established exclusively through outbound 
encrypted tunnels, eliminating the need for 
inbound port exposure on the client sides 
without architectural or configurations modifications.

\subsection{Deployment Scenario}
\label{subsec4.2}

The evaluation was structured to progressively validate 
functionality, multi-tenant isolation, and operational stability 
across three independent organizational environments.

Three tenants (Organization A, Organization B, and Organization 
C) were configured within the platform, each assigned a dedicated 
edge device and a fully isolated data space. As shown in 
Figure~\ref{fig:tenant-id}, each organization was provisioned 
with an independent identifier and a specific set of security 
modules. No information was shared across tenant boundaries, 
enabling assessment of data separation and organizational 
independence under concurrent operation.

\begin{figure} [htb!]
    \centering
    \includegraphics[width=1\linewidth]{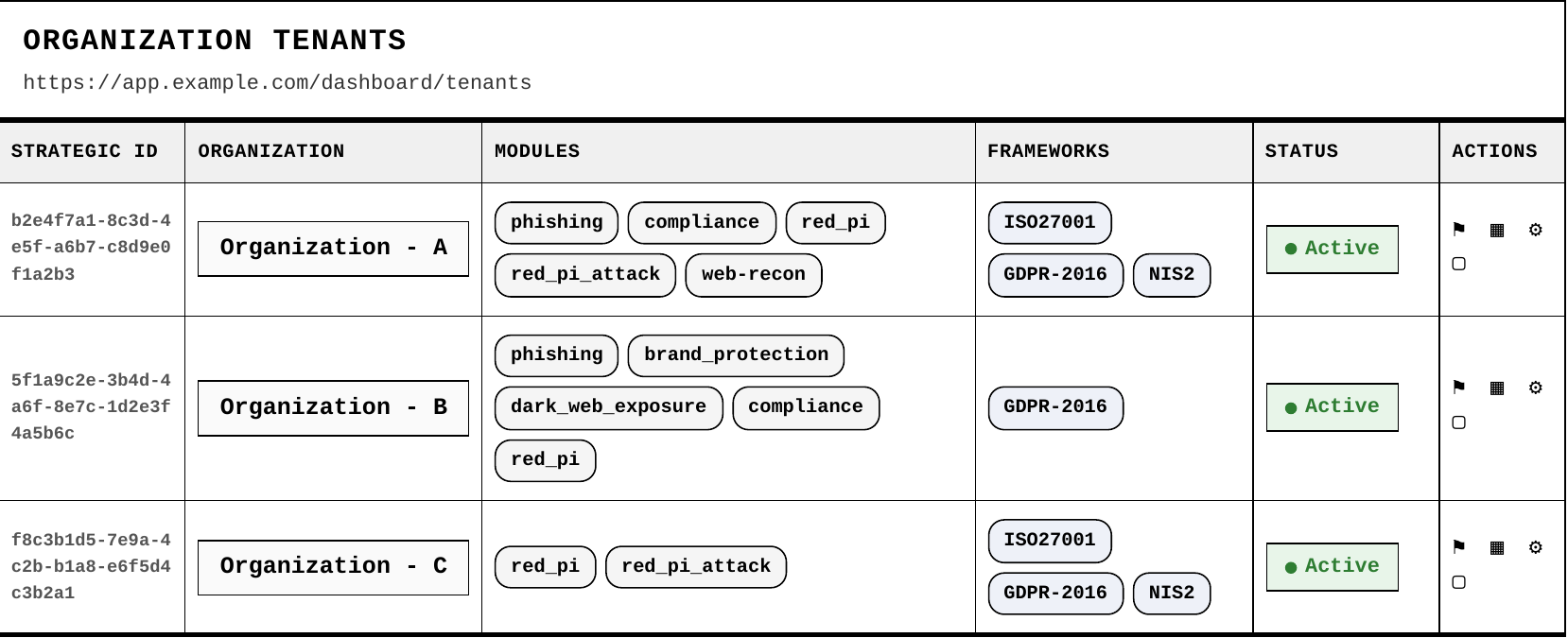}
    \caption{Organization Tenants}
    \label{fig:tenant-id}
\end{figure}

Figure~\ref{fig:devices} shows the three edge 
devices registered within the platform during 
the evaluation, displaying their operational 
status, hardware architecture, assigned tenant, 
and network addressing. The heterogeneous nature 
of the deployment is reflected in the mix of 
ARM and x86 architectures, as well as the 
combination of online and offline states 
observed during the evaluation period.

\begin{figure} [htb!]
    \centering
    \includegraphics[width=1\linewidth]{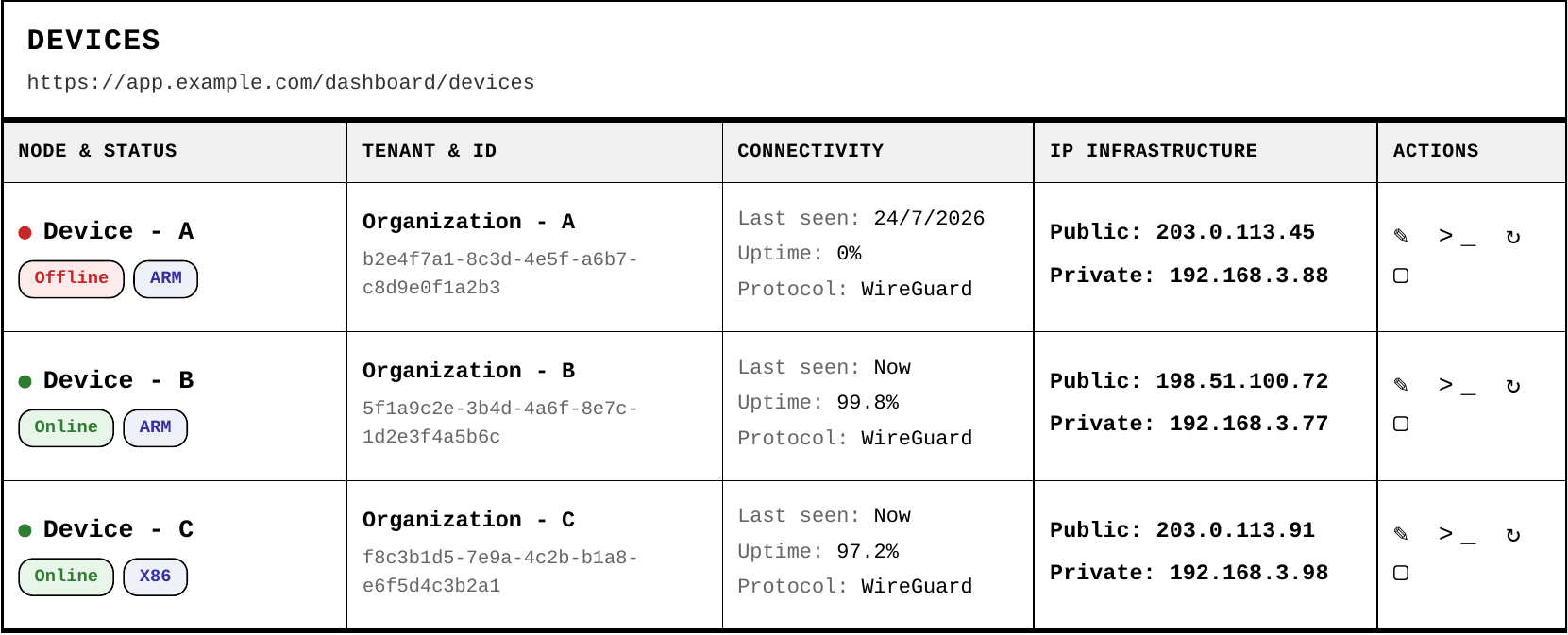}
    \caption{Devices Panel}
    \label{fig:devices}
\end{figure}

Two distinct user roles were exercised during the evaluation. An 
administrator held full access to platform configuration, tenant 
management, and backend services. A read-only role, assigned to 
client-facing users within each organization, was limited to 
viewing dashboards and execution results, with no visibility into 
the administrative panel or other tenants' data. This separation 
confirmed that the RBAC model enforces clear boundaries between 
technical and supervisory profiles.

As illustrated in Figure~\ref{fig:menu-admin-user}, the 
navigation menus exposed to each role differ substantially. The 
read-only user accesses only Dashboard, Modules, Alerts, 
Incidents, and Support. The administrator has additional access 
to Tenants, Devices, Hosts, SMTP configuration, and debug 
capabilities, none of which are visible to client-facing users.

\begin{figure}[htb!]
    \centering
    \includegraphics[width=0.6\linewidth]{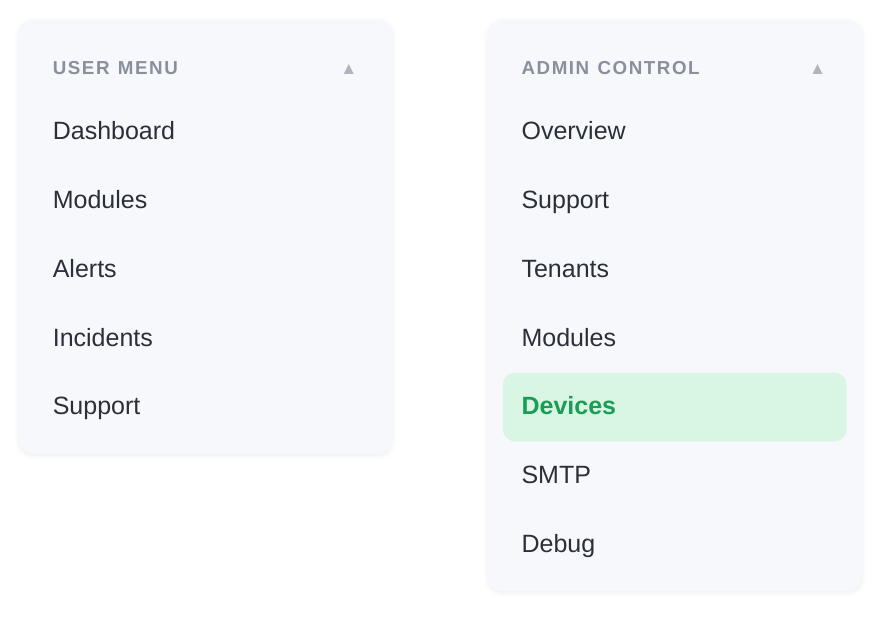}
    \caption{User Menu  vs Admin Control}
    \label{fig:menu-admin-user}
\end{figure}

All nine security tasks were deployed from the outset and 
executed automatically through scheduled daily routines, without 
manual operator intervention. Each task ran once per day per edge 
device over thirty days, yielding 30 executions per task per 
device and 90 executions per task across three endpoints, for a 
total of 810 automated executions. This schedule was designed to 
assess operational stability and execution consistency over a 
period representative of a real-world deployment month.

\subsection{Task Execution Workflow}
\label{subsec4.3}

Security tasks were selected and mapped to the 
MITRE ATT\&CK framework \cite{mitre} to provide 
a structured and reproducible reference for 
evaluating the platform's orchestration 
capabilities across representative attack 
surface coverage phases, including 
reconnaissance, enumeration, and credential 
access. All executions were performed 
automatically through the scheduled daily 
routines defined in the deployment scenario, 
with no manual operator intervention

\begin{figure} [h!]
    \centering 
    \includegraphics[width=1\linewidth]{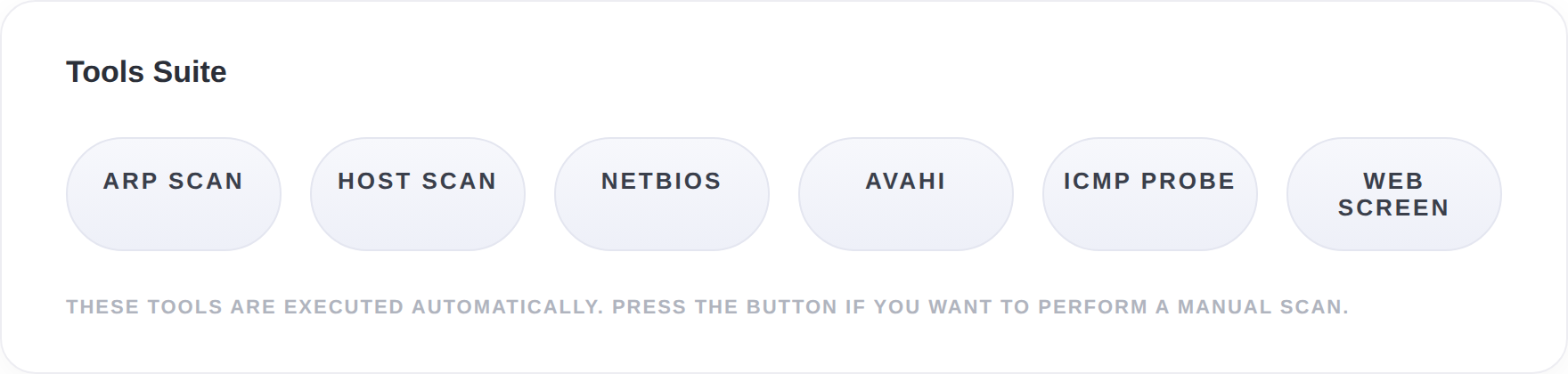}
    \caption{Recon Tools Suite}
    \label{fig:modulo-recon}
\end{figure}

\begin{figure} [h!]
    \centering
    \includegraphics[width=1\linewidth]{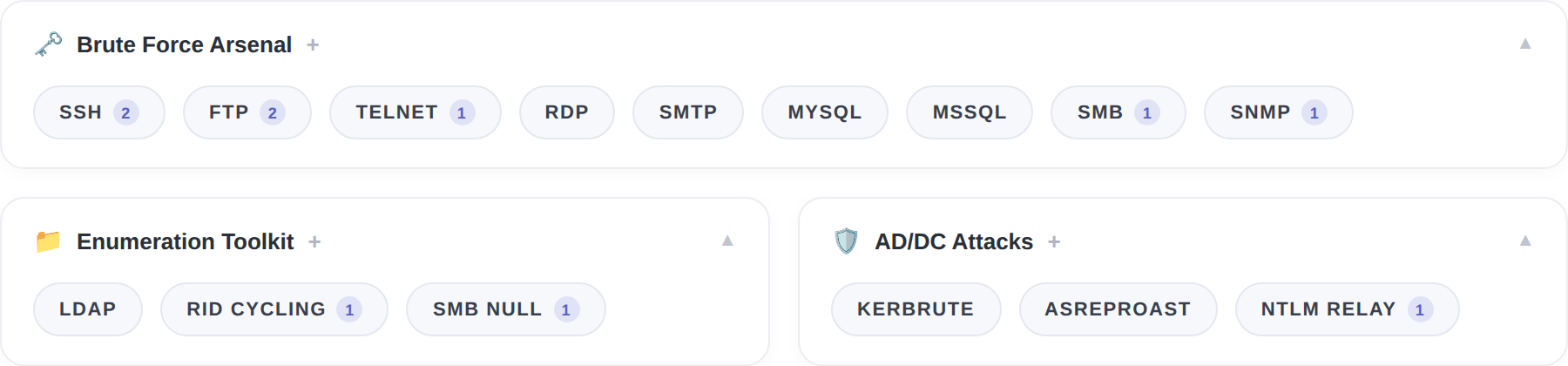}
    \caption{Attack Arsenal}
    \label{fig:modulos-attack}
\end{figure}

Figure~\ref{fig:modulo-recon} ~\ref{fig:modulos-attack} illustrates the 
security module categories available within 
the platform, organized into reconnaissance 
tools, brute force capabilities, enumeration 
toolkit, and AD/DC attack modules. The nine 
tasks selected for this evaluation represent 
a representative subset of the available 
toolset, covering reconnaissance, discovery, 
collection, initial access, and credential 
access phases of the MITRE ATT\&CK framework.

Table~\ref{tab:tasks} summarizes the nine tasks 
executed during the evaluation, their associated 
MITRE technique identifiers, and the number of 
executions performed across the three edge 
devices.

\begin{table}[htbp]
\centering
\caption{Security tasks
and executions per edge device}
\label{tab:tasks}
\footnotesize
\begin{tabular}{p{3.4cm} p{1.8cm} p{0.6cm} 
p{0.6cm} p{0.6cm} p{0.6cm}}
\hline
\hline
Task & MITRE ID & A & B & C & Total \\
\hline
Host IP \& ARP Discovery  & T1595.001 & 30 & 30 & 30 & 90 \\
Port \& Service Enum.     & T1595.002 & 30 & 30 & 30 & 90 \\
AD Enumeration            & T1087.002 & 30 & 30 & 30 & 90 \\
Default Credential Abuse  & T1078.001 & 30 & 30 & 30 & 90 \\
SMB Share Discovery       & T1135     & 30 & 30 & 30 & 90 \\
SNMP MIB Dump             & T1602.001 & 30 & 30 & 30 & 90 \\
Kerberos Bruteforce       & T1110     & 30 & 30 & 30 & 90 \\
AS-REP Roasting           & T1558.004 & 30 & 30 & 30 & 90 \\
Responder / LLMNR         & T1557.001 & 30 & 30 & 30 & 90 \\
\hline
\textbf{Total} & & \textbf{270} & \textbf{270} 
& \textbf{270} & \textbf{810} \\
\hline
\hline
\end{tabular}
\end{table}

Each task was implemented as an isolated 
container image invoked by the execution agent 
on the edge device upon receiving an authorized 
request from the automation engine. Upon 
completion, results were returned in structured 
JSON format to the control plane for storage 
and visualization.

\subsection{Data Collection and Monitoring}
\label{subsec4.4}

During the evaluation, \textit{Argos} 
continuously collected structured data for each 
task execution. Upon completion, each task 
generated a log entry containing a unique 
execution identifier, a timestamp, an execution 
status indicator, and the output data returned 
by the edge device in JSON format. This output 
was parsed by the automation engine and 
persistently stored in the centralized relational 
database, scoped to the corresponding tenant 
identifier. Both successful executions and 
failure conditions were recorded, enabling 
comprehensive visibility into platform behavior 
over the evaluation period.

The operational status of each edge device was 
monitored through a periodic heartbeat mechanism, 
reporting device availability to the control 
plane every five minutes. In the event of a 
missed report, the platform automatically 
generated a device unavailability alert, enabling 
detection of connectivity disruptions without 
requiring direct access to the client environment. 
During the evaluation period, isolated 
connectivity incidents were recorded across the 
three deployment sites, attributable to 
external factors such as power outages, voltage 
fluctuations, and internet service interruptions 
at the client side. These events were 
automatically detected and logged by the 
platform, and did not result in data loss or 
cross-tenant interference.

Collected data was exposed to users through the 
web-based dashboard according to their assigned 
role and tenant context. Read-only users could 
access execution histories, task outcomes, and 
aggregated result views restricted to their 
organization. Administrator users had access to 
additional platform information including device 
configuration and system settings. At no point 
during the evaluation was information shared or 
visible across tenant boundaries.

\begin{figure}[htb!]
    \centering
    \includegraphics[width=0.9\linewidth]{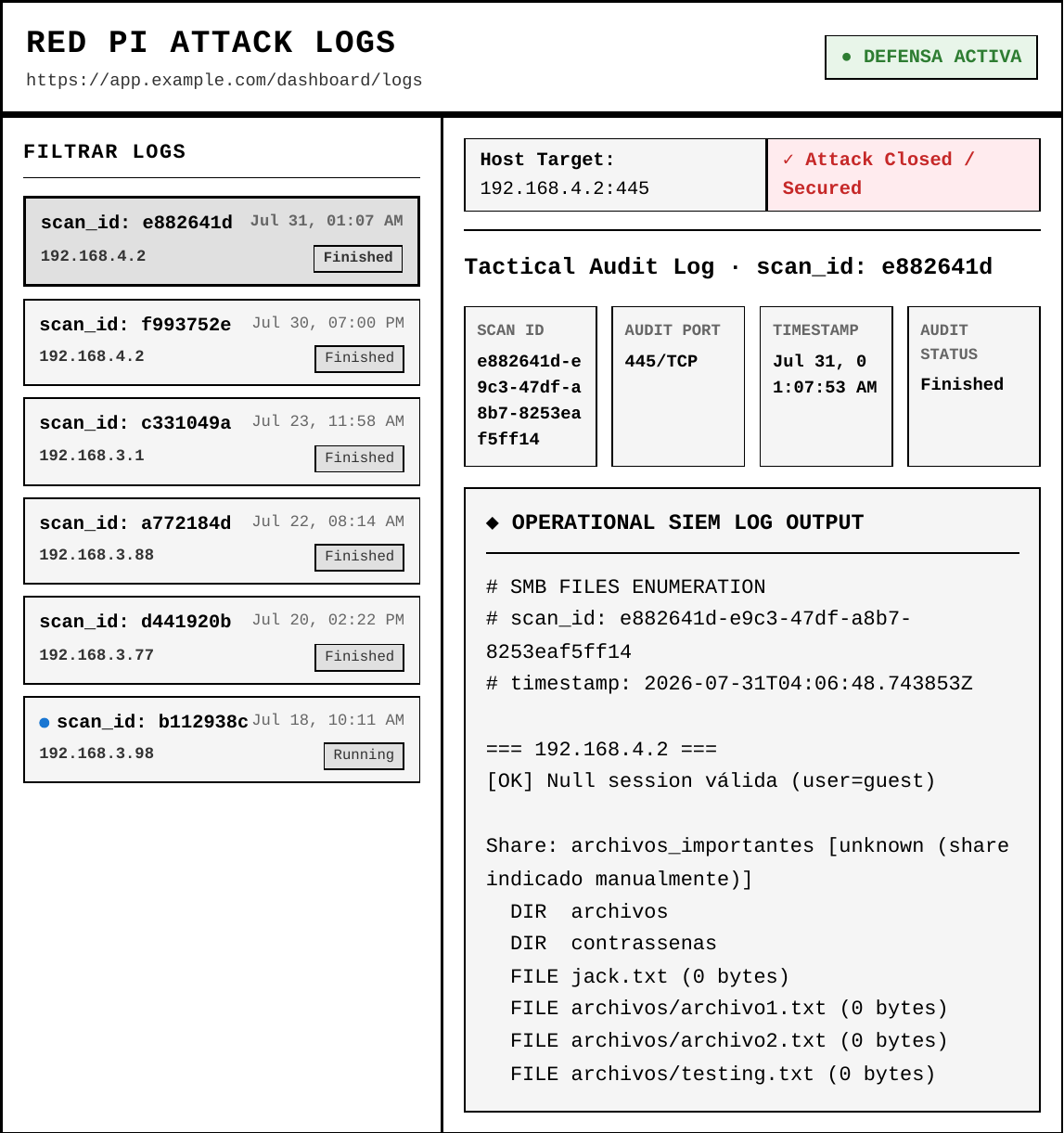}
    \caption{Example task execution log: SMB null 
    session enumeration result }
    \label{fig:execution_log}
\end{figure}

Figure~\ref{fig:execution_log} illustrates a representative 
execution log generated during the evaluation. Each entry 
includes a unique scan identifier, target host and port, 
timestamp, and audit status determined automatically by the 
platform. The structured SIEM output captures raw findings 
returned by the edge device. In this case, an SMB null session 
enumeration exposed three shared resources on the target host, 
classified by the platform as a security breach. This example 
shows that \textit{Argos} not only orchestrates task execution 
but also parses and classifies results, surfacing actionable 
findings directly through the dashboard.

In addition to per-execution logs, the platform generates 
structured security reports summarizing findings collected across 
tasks and edge devices within a tenant. These reports provide an 
auditable record of discovered hosts, identified services, and 
security findings, supporting the traceability and accountability 
requirements of authorized auditing operations.

\subsection{Evaluation Criteria}
\label{subsec4.5}

The evaluation of \textit{Argos} was conducted 
using eight operational criteria designed to 
validate the principal architectural guarantees 
of the platform and directly address the 
structural gaps identified in Section~\ref{sec:related}.

\textbf{Role-based access enforcement.} The RBAC model must 
enforce strict operational boundaries between user profiles. 
Administrator users have access to platform configuration, tenant 
management, and device settings. Read-only users are restricted 
to their organization's dashboards and results, with no 
visibility into administrative functions or other tenants' data.

\textbf{Tenant-device binding.} Task execution requests issued 
by a tenant must be delivered exclusively to the edge device 
associated with that tenant. Cross-tenant device routing is not 
permitted. A task initiated within Organization A must never be 
dispatched to a device belonging to Organization B or C.

\textbf{Result isolation.} Execution results, logs, and audit 
data collected by each edge device must remain strictly scoped 
to their originating tenant. Users of one organization must not 
be able to access data generated within another organization's 
scope, regardless of their assigned role.

\textbf{Cross-architecture compatibility.} The platform must 
orchestrate task execution consistently across heterogeneous 
hardware. Edge devices based on both ARM and x86 processors were 
included in the evaluation to confirm that the containerized 
execution model operates correctly regardless of the underlying 
hardware platform.

\textbf{Latency resilience.} Network latency introduced by 
geographically distributed deployments must not affect execution 
correctness. The platform's use of lightweight HTTPS exchanges 
rather than persistent sessions is designed to make execution 
outcomes independent of latency variability across regions.

\textbf{Dashboard result visibility.} Execution results must be 
correctly surfaced through the web dashboard for each tenant upon 
task completion. This validates end-to-end traceability from task 
dispatch to result visualization, ensuring findings are accessible 
to authorized users in a timely and organized manner.

\textbf{Automation effectiveness.} The platform supports both 
manual and scheduled task execution. Manual triggering was used 
for functional validation, while the automated daily scheduler 
assessed platform stability over the evaluation period. The 
scheduling engine is configurable by the administrator, allowing 
each organization to define custom execution routines adapted to 
their infrastructure and operational requirements.

\textbf{Execution repeatability.} Tasks must produce consistent 
and reliable outcomes across repeated executions and heterogeneous 
deployment environments. Each of the nine tasks was executed once 
per day per device over thirty days, enabling assessment of 
temporal stability and cross-environment consistency under real 
operational conditions.

Together, these criteria provide a comprehensive operational 
assessment of \textit{Argos}, directly addressing identified 
gaps in existing C2 platforms with respect to governance, 
isolation, automation, and long-term edge deployment.

\section{Results and Discussion}
\label{sec:results}

\subsection{Results}
\label{subsec5.1}

The RBAC model was validated by exercising two 
distinct user profiles throughout the evaluation. 
The administrator role provided full access to 
platform configuration, tenant management, and 
device settings, while the read-only role was 
restricted exclusively to dashboards and results 
scoped to their own tenant. As illustrated in 
Figure~\ref{fig:menu-admin-user}, the navigation 
menus exposed to each role differ substantially, 
confirming that the platform enforces strict 
functional boundaries between technical and 
supervisory profiles.

Tenant and device isolation were validated 
through controlled cross-tenant access attempts 
and cross-edge visibility checks. As summarized 
in Table~\ref{tab:tenant_isolation}, no 
isolation violations were observed, confirming 
that tenant-device binding and result scoping 
operate correctly under concurrent 
multi-organizational deployment.

\begin{table}[htb]
\centering
\caption{Multi-Tenant Isolation Validation 
Results per Organization}
\label{tab:tenant_isolation}
\footnotesize
\begin{tabular}{p{3.5cm} p{1.1cm} p{1.1cm} 
p{1.1cm} p{1.1cm} p{1.1cm}}
\hline
\hline
\textbf{Operation Type} & \textbf{Org A} & 
\textbf{Org B} & \textbf{Org C} & 
\textbf{Total} & \textbf{Violations} \\
\hline
File creation    & 30 & 30 & 30 & 90 & 0 \\
File modification & 30 & 30 & 30 & 90 & 0 \\
File deletion    & 30 & 30 & 30 & 90 & 0 \\
Cross-tenant access & 10 & 10 & 10 & 30 & 0 \\
Cross-edge check & 10 & 10 & 10 & 30 & 0 \\
\hline
\textbf{Total}   & \textbf{110} & \textbf{110} 
& \textbf{110} & \textbf{330} & \textbf{0} \\
\hline
\hline
\end{tabular}
\end{table}

All 810 task executions were performed 
automatically through the platform's scheduled 
daily automation engine, without manual operator 
intervention during the repeated execution phase. 
Each of the nine security tasks was executed once 
per day per edge device over thirty days, enabling 
assessment of both temporal stability and 
cross-environment consistency.

Table~\ref{tab:repeat_orga}, 
\ref{tab:repeat_orgb}, and 
\ref{tab:repeat_orgc} detail the execution 
results per task for each independent 
organizational deployment.

\begin{table}[htbp]
\centering
\caption{Task Execution Results ---
Organization A (ARM)}
\label{tab:repeat_orga}
\footnotesize
\begin{tabular}{p{3.5cm} p{2cm} p{1cm} p{0.6cm} p{0.85cm} p{1.6cm}}
\hline\hline
\textbf{Task} & \textbf{MITRE ID} & \textbf{Succ.} & \textbf{Fail.} & \textbf{Rate} & \textbf{95\% CI} \\
\hline
Host IP \& ARP Discovery    & T1595.001 & 30 & 0 & 100.0\% & [88.6, 100.0] \\
Port \& Service Enum.       & T1595.002 & 28 & 2 & 93.3\%  & [78.7, 98.2] \\
AD Enumeration              & T1087.002 & 29 & 1 & 96.7\%  & [83.3, 99.4] \\
Default Credential Abuse    & T1078.001 & 30 & 0 & 100.0\% & [88.6, 100.0] \\
SMB Share Discovery         & T1135     & 30 & 0 & 100.0\% & [88.6, 100.0] \\
SNMP MIB Dump               & T1602.001 & 28 & 2 & 93.3\%  & [78.7, 98.2] \\
Kerberos Bruteforce         & T1110     & 27 & 3 & 90.0\%  & [74.4, 96.5] \\
AS-REP Roasting             & T1558.004 & 30 & 0 & 100.0\% & [88.6, 100.0] \\
Responder / LLMNR           & T1557.001 & 26 & 4 & 86.7\%  & [70.3, 94.7] \\
\hline
\textbf{Total} & & \textbf{258} & \textbf{12} & \textbf{95.6\%} & \textbf{[92.4, 97.4]} \\
\hline\hline
\end{tabular}
\end{table}

\begin{table}[htbp]
\centering
\caption{Task Execution Results ---
Organization B (ARM)}
\label{tab:repeat_orgb}
\footnotesize
\begin{tabular}{p{3.5cm} p{2cm} p{1cm} p{0.6cm} p{0.85cm} p{1.6cm}}
\hline\hline
\textbf{Task} & \textbf{MITRE ID} & \textbf{Succ.} & \textbf{Fail.} & \textbf{Rate} & \textbf{95\% CI} \\
\hline
Host IP \& ARP Discovery    & T1595.001 & 30 & 0 & 100.0\% & [88.6, 100.0] \\
Port \& Service Enum.       & T1595.002 & 29 & 1 & 96.7\%  & [83.3, 99.4] \\
AD Enumeration              & T1087.002 & 29 & 1 & 96.7\%  & [83.3, 99.4] \\
Default Credential Abuse    & T1078.001 & 29 & 1 & 96.7\%  & [83.3, 99.4] \\
SMB Share Discovery         & T1135     & 30 & 0 & 100.0\% & [88.6, 100.0] \\
SNMP MIB Dump               & T1602.001 & 29 & 1 & 96.7\%  & [83.3, 99.4] \\
Kerberos Bruteforce         & T1110     & 28 & 2 & 93.3\%  & [78.7, 98.2] \\
AS-REP Roasting             & T1558.004 & 30 & 0 & 100.0\% & [88.6, 100.0] \\
Responder / LLMNR           & T1557.001 & 27 & 3 & 90.0\%  & [74.4, 96.5] \\
\hline
\textbf{Total} & & \textbf{261} & \textbf{9} & \textbf{96.7\%} & \textbf{[93.8, 98.2]} \\
\hline\hline
\end{tabular}
\end{table}

\begin{table}[htbp]
\centering
\caption{Task Execution Results ---
Organization C (x86)}
\label{tab:repeat_orgc}
\footnotesize
\begin{tabular}{p{3.5cm} p{2cm} p{1cm} p{0.6cm} p{0.85cm} p{1.6cm}}
\hline\hline
\textbf{Task} & \textbf{MITRE ID} & \textbf{Succ.} & \textbf{Fail.} & \textbf{Rate} & \textbf{95\% CI} \\
\hline
Host IP \& ARP Discovery    & T1595.001 & 29 & 1 & 96.7\%  & [83.3, 99.4] \\
Port \& Service Enum.       & T1595.002 & 30 & 0 & 100.0\% & [88.6, 100.0] \\
AD Enumeration              & T1087.002 & 29 & 1 & 96.7\%  & [83.3, 99.4] \\
Default Credential Abuse    & T1078.001 & 29 & 1 & 96.7\%  & [83.3, 99.4] \\
SMB Share Discovery         & T1135     & 30 & 0 & 100.0\% & [88.6, 100.0] \\
SNMP MIB Dump               & T1602.001 & 29 & 1 & 96.7\%  & [83.3, 99.4] \\
Kerberos Bruteforce         & T1110     & 29 & 1 & 96.7\%  & [83.3, 99.4] \\
AS-REP Roasting             & T1558.004 & 30 & 0 & 100.0\% & [88.6, 100.0] \\
Responder / LLMNR           & T1557.001 & 29 & 1 & 96.7\%  & [83.3, 99.4] \\
\hline
\textbf{Total} & & \textbf{264} & \textbf{6} & \textbf{97.8\%} & \textbf{[95.2, 99.0]} \\
\hline\hline
\end{tabular}
\end{table}

Across the three deployments, 783 successful
executions were recorded out of 810, yielding a
global success ratio of $\hat{p} = 0.9667$ and a
corresponding failure rate of 3.33\%, computed
through the maximum-likelihood estimator of
Equation~\eqref{eq:phat}. Because the estimated
success probability lies close to unity, interval
estimation was performed using the Wilson score
interval of Equation~\eqref{eq:wilson} rather than
the symmetric Wald approximation. At a 95\%
confidence level ($z = 1.96$), the resulting
interval for the failure rate is $[2.30\%,
4.81\%]$, confirming that the observed reliability
is tightly bounded rather than an artifact of a
single favourable run. The majority of observed
failures were attributable to external
connectivity factors, including internet service
interruptions, power outages, voltage
fluctuations, and tunnel disconnections at the
client side, rather than to architectural
deficiencies in the orchestration model. No
systematic execution degradation was observed over
the thirty-day evaluation period, confirming the
platform's operational stability under sustained
automated workloads.

Edge devices A and B operated on ARM-based
hardware while device C ran on an x86 virtualized
environment. Comparing results across
architectures, Organization C (x86) achieved the
lowest failure rate (2.2\%) versus Organizations A
and B (4.4\% and 3.3\% respectively). To determine
whether this difference reflects a genuine
architectural effect or ordinary sampling
variation, the executions were aggregated into a
two-way contingency table by processor family and
outcome, as reported in
Table~\ref{tab:arch_contingency}.

\begin{table}[htbp]
\centering
\caption{Contingency of execution outcomes by
hardware architecture}
\label{tab:arch_contingency}
\footnotesize
\begin{tabular}{p{3.0cm} p{1.1cm} p{1.0cm} p{0.8cm} p{0.9cm} p{1.5cm}}
\hline\hline
\textbf{Architecture} & \textbf{Success} & \textbf{Failure} & \textbf{Total} & \textbf{Rate} & \textbf{95\% CI} \\
\hline
ARM (Org A + B) & 519 & 21 & 540 & 96.1\% & [94.1, 97.4] \\
x86 (Org C)     & 264 &  6 & 270 & 97.8\% & [95.2, 99.0] \\
\hline
\textbf{Total}  & 783 & 27 & 810 & 96.7\% & [95.2, 97.7] \\
\hline\hline
\end{tabular}
\end{table}

Applying Pearson's chi-squared test of
Equation~\eqref{eq:chi2} to
Table~\ref{tab:arch_contingency} yields
$\chi^{2} = 1.55$ with one degree of freedom and
$p = 0.21$. As the statistic falls well below the
critical value of $3.84$ at the $0.05$
significance level, the null hypothesis of
architecture-independent reliability cannot be
rejected. The same conclusion is obtained under Fisher's
exact test ($p = 0.30$) and under the Yates
continuity-corrected statistic ($\chi^{2} = 1.08$,
$p = 0.30$); all expected cell counts exceed five,
so the chi-squared assumptions are satisfied. The observed difference is
therefore not statistically significant, and the
containerized execution model can be regarded as
hardware-agnostic for standard network security
tasks. The residual differential concentrates in
the Responder/LLMNR task, where ARM devices
accumulated four and three failures respectively
while x86 recorded only one; this behaviour is
attributable to network-interface parameter
modification requirements that are more disruptive
on physical ARM hardware than on virtualized x86
environments, rather than to a systematic
architectural limitation.

Figure~\ref{fig:latency} summarizes 
the average round-trip latency observed between 
the AWS cloud-hosted control plane and each 
edge device during the evaluation period. 
Latency was measured via the periodic heartbeat 
mechanism, which reports device availability 
every five minutes.

% \begin{table}[htbp]
% \centering
% \caption{Observed Latency Between Control 
% Plane and Edge Devices}
% \label{tab:latency_analysis}
% \footnotesize
% \begin{tabular}{p{2cm} p{2.5cm} p{1.8cm} p{3cm}}
% \hline
% \hline
% \textbf{Node} & \textbf{Region} & 
% \textbf{Avg Latency} & 
% \textbf{Operational Effect} \\
% \hline
% Org A & Europe (same region) & $\sim$20 ms & 
% Negligible delay \\
% Org B & South America       & $\sim$120 ms & 
% No instability \\
% Org C & South America       & $\sim$50 ms  & 
% No instability \\
% \hline
% \hline
% \end{tabular}
% \end{table}

\begin{figure}[htb!]
    \centering
    \includegraphics[width=0.8\linewidth]{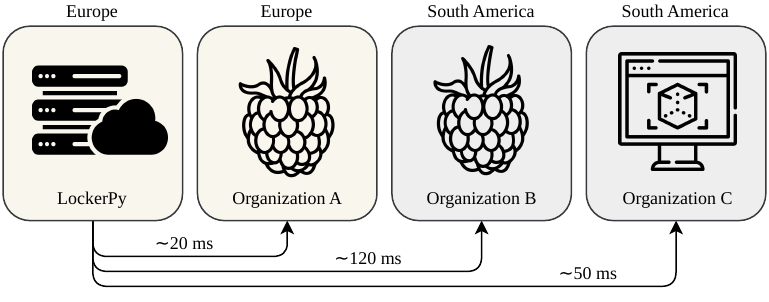}
    \caption{Observed Latency Between Control Plane and Edge Devices}
    \label{fig:latency}
\end{figure}

Despite the latency variability across 
deployments, no execution instability 
attributable to network delay was observed 
in any of the three organizational 
environments. The platform's reliance on 
lightweight HTTPS request-response exchanges 
with structured JSON payloads means that 
increased latency affects task initiation 
timing but does not degrade workflow 
integrity or execution correctness. 
Organization C, despite being located in
South America, benefited from an enhanced
international connectivity infrastructure,
resulting in a round-trip latency well below
that of typical intercontinental links. This confirms that 
\textit{Argos}'s communication model is 
resilient to heterogeneous network conditions 
and does not impose geographic constraints 
on edge device deployment.

Figure~\ref{fig:dashboard_recon} illustrates
the reconnaissance dashboard for Organization
A, showing aggregated metrics including 17
discovered hosts, 10 active endpoints, 594
total scan operations dispatched, and 17
open ports identified, all scoped exclusively
to that organization's edge device. In this
context a scan operation denotes an individual
scan action issued to the edge device and is
accumulated over the full thirty-day evaluation
period; since each scheduled reconnaissance task
probes multiple hosts and ports, the number of
scan operations is expectedly larger than the
count of scheduled task executions.

\begin{figure}[t!]
    \centering
    \includegraphics[width=1\linewidth]{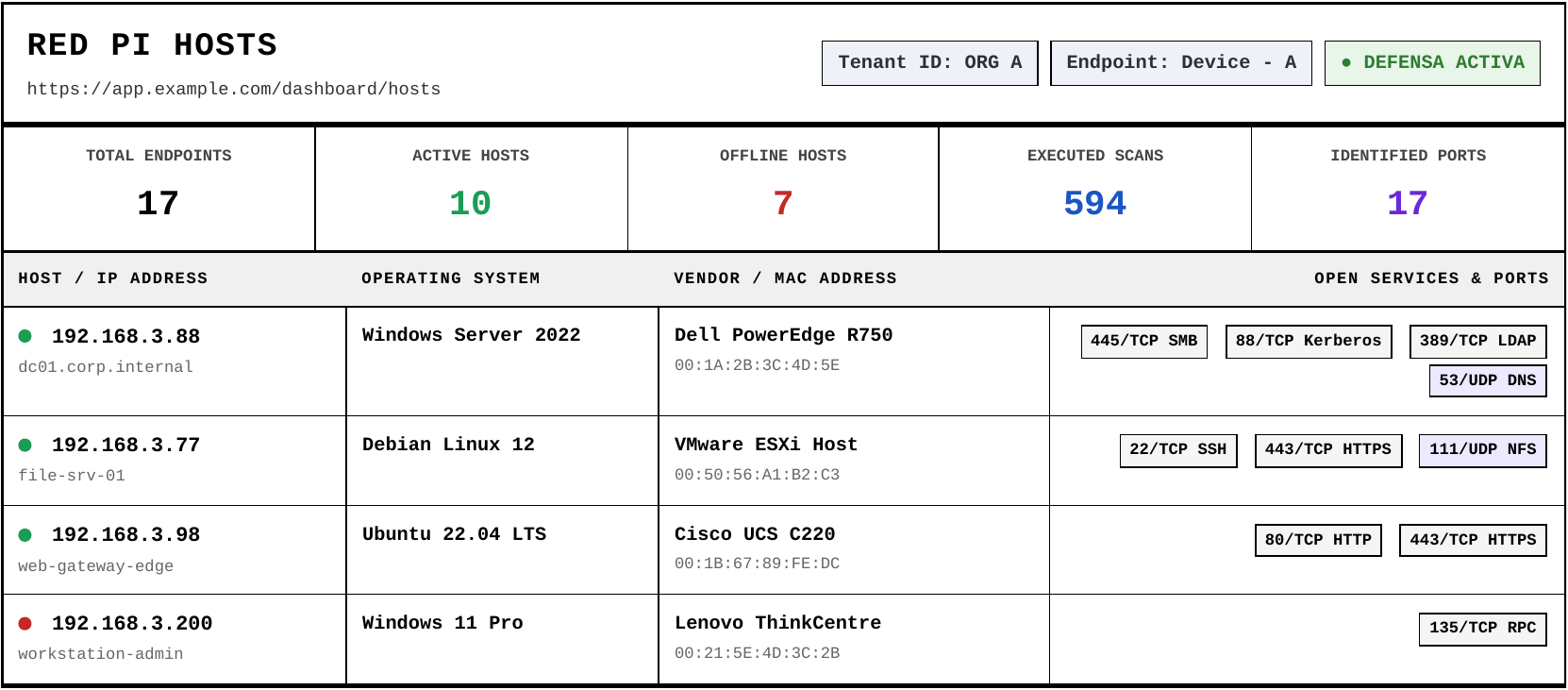}
    \caption{Hosts dashboard for 
    Organization A.}
    \label{fig:dashboard_recon}
\end{figure}

\textit{Argos} additionally generates 
structured infrastructure audit reports 
consolidating findings across tasks and 
edge devices within a tenant. As illustrated 
in Figure~\ref{fig:report_recon}, these 
reports provide an executive summary of 
discovered assets, active hosts, total scans, 
and exposed services, complemented by a 
discovery timeline and service breakdown. 
This reporting capability supports the 
auditability and traceability requirements 
of authorized security operations, enabling 
non-technical stakeholders to access 
high-level security insights in a structured 
and exportable format.

\begin{figure}[htb!]
    \centering
    \includegraphics[width=0.8\linewidth]{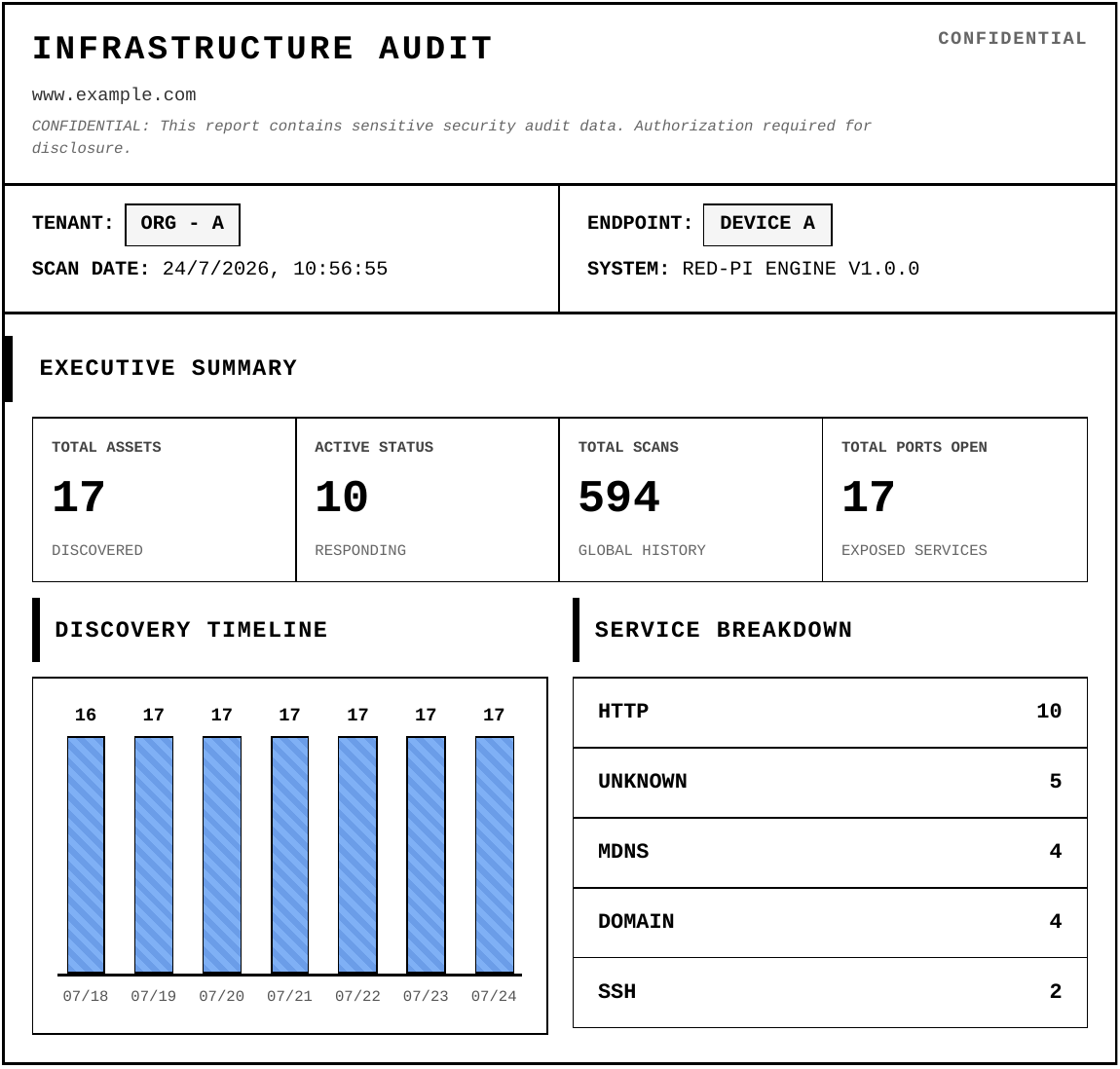}
    \caption{Infrastructure audit report 
    generated for 
    Organization A.}
    \label{fig:report_recon}
\end{figure}

\subsection{Discussion}
\label{subsec5.2}

The evaluation results demonstrate that \textit{Argos} 
maintains stable orchestration across distributed and 
heterogeneous environments. The global success rate was 96.7\% 
(95\% Wilson confidence interval $[95.19\%, 97.70\%]$) across 
810 automated executions over thirty days without manual operator 
intervention. The statistical analysis confirms two central 
claims: the reliability estimate is tightly bounded as 
established through Equation~\eqref{eq:wilson}, and the 
containerized execution model is hardware-agnostic, since the 
difference in failure rates between ARM and x86 deployments is 
not statistically significant (Table~\ref{tab:arch_contingency}). 
These findings are consistent with Property~\ref{prop:latency}, 
which predicts that execution correctness is preserved regardless 
of the geographic distribution of edge devices.

The results directly address the structural gaps identified in 
Section~\ref{sec:related}. Role-based access control was enforced 
consistently, tenant and device isolation operated correctly 
across all 330 validation operations, and the automation engine 
executed all scheduled tasks autonomously. End-to-end 
traceability from task dispatch to dashboard visualization and 
report generation was confirmed throughout the evaluation period.

The primary operational risk identified was tunnel stability 
rather than architectural limitations. Most observed failures 
were attributable to external factors at the client side, 
including power outages, ISP interruptions, and voltage 
fluctuations. This highlights that communication channel 
reliability is partially dependent on the client's 
infrastructure quality, which lies outside the platform's 
direct control.

\subsection{Limitations}
\label{subsec:limitations}

Several limitations must be acknowledged based on the evaluation 
conducted.

Regarding edge device hardware, the single-board computer used 
in physical deployments is suitable for environments with up to 
approximately 200 discoverable hosts. As the number of tasks or 
network scope increases, the computational demands of concurrent 
containerized execution may exceed the capacity of this hardware 
class, requiring migration to a more capable device such as a 
mini PC, which remains fully compatible with the platform's 
execution model.

Connectivity stability at the edge device is partially dependent 
on the client organization's network infrastructure. Power 
outages, ISP interruptions, and voltage fluctuations can cause 
tunnel disconnections that result in missed executions, as 
observed during the evaluation. This represents an environmental 
constraint outside the platform's direct control.

The evaluation was conducted across three edge devices with 
client environments of up to 180 hosts. Deployments involving 
50 to 100 simultaneous edge devices, or environments exceeding 
200 hosts per node, were not assessed. Formal scalability testing 
under increased endpoint density remains a direction for future 
work.

Finally, while the platform currently integrates AI capabilities 
within the automation engine, the incorporation of autonomous 
agents capable of adaptive decision-making during execution was 
not fully explored. This constitutes a natural evolution of the 
platform's orchestration model and a promising direction for 
future development.

\section{Conclusion and Future Work}
\label{sec:conclusion}

This work presented \textit{Argos}, a multi-tenant 
command-and-control platform for transparent, automated, and 
long-term orchestration of distributed security auditing devices. 
Unlike traditional offensive C2 frameworks that prioritize 
stealth and operator-driven interaction, \textit{Argos} focuses 
on structured automation, role-based governance, and tenant 
isolation for legitimate operational environments.

The evaluation demonstrated stable orchestration across three 
geographically distributed edge devices with a global success 
rate of 96.7\% over 810 automated executions during thirty days. 
Role-based access control, tenant-device binding, and result 
isolation were validated without a single violation across 330 
controlled operations. The automated scheduling engine operated 
without manual intervention throughout, confirming the platform's 
suitability for long-term autonomous security auditing.

As future work, scalability testing involving 50 or more 
simultaneous edge devices will assess behavior under higher 
concurrent workloads. Evaluation under adverse network conditions, 
including simulated denial-of-service scenarios and persistent 
packet loss, will strengthen the resilience assessment. The 
integration of autonomous AI agents within security task workflows 
represents a natural evolution of the orchestration model, 
enabling adaptive decision-making beyond current rule-based 
automation.

\newpage

\bibliographystyle{elsarticle-num} 
\bibliography{refrences}

\end{document}